\documentclass[acmtocl]{robtrans}%
\pdfoutput=1

\usepackage{dashrule}
\usepackage{proof-dashed}
\usepackage{amssymb}
\usepackage{amsmath}
\usepackage{verbatim}
\usepackage{stmaryrd}

\newtheorem{theorem}{Theorem}
\newtheorem{lemma}{Lemma}
\newtheorem{corollary}{Corollary}

\usepackage{tikz}
\usetikzlibrary{arrows}

\newcommand{\seq}[2]{{#1} \longrightarrow {#2} \mathstrut}

\newcommand{\susp}[1]{\langle{#1}\rangle}

\newcommand{\efoc}[3]{{#1}; {#2} \vdash {#3}}
\newcommand{\rfoc}[2]{{#1} \vdash [{#2}] \mathstrut}
\newcommand{\lfoc}[3]{{#1}; [{#2}] \vdash {#3} \mathstrut}
\newcommand{\ifoc}[3]{{#1}; {#2} \vdash {#3} \mathstrut}
\newcommand{\efoct}[4]{{#1}; {#2} \vdash {#3} : {#4}}
\newcommand{\rfoct}[3]{{#1} \vdash {#2} : [{#3}] \mathstrut}
\newcommand{\lfoct}[4]{{#1}; [{#3}] \vdash {#2} : {#4} \mathstrut}
\newcommand{\ifoct}[4]{{#1}; {#2} \vdash {#3} : {#4} \mathstrut}
\newcommand{\stable}[1]{{#1}\,\mathit{stable} \mathstrut}

\newcommand{\llangle}{\llbracket}
\newcommand{\rrangle}{\rrbracket}

\DeclareMathOperator{\with}{\&}

\newcommand{\dsrt}[1]{\mathsf{thunk}\,{#1}}
\newcommand{\dslt}[2]{{#1}.{#2}}      
\newcommand{\usrt}[1]{\{{#1}\}}       
\newcommand{\uslt}[1]{\mathsf{pm}\,{#1}}       

\newcommand{\etart}[1]{\langle{#1}\rangle} 
\newcommand{\etalt}[2]{\susp{#1}.{#2}}     

\newcommand{\rft}[1]{\mathsf{ret}\,{#1}} 
\newcommand{\lft}[2]{{#1} \circ {#2}} 

\title{Structural focalization}
\author{Robert J. Simmons}

\begin{abstract}
Focusing, introduced by Jean-Marc Andreoli in the context of classical
linear logic [Andreoli 1992], 
defines a normal form for sequent calculus derivations that 
cuts down on the number of possible derivations by eagerly applying invertible
rules and grouping sequences of non-invertible rules. 
A focused sequent calculus is defined relative to some non-focused
sequent calculus; {\it focalization} is the property that
every non-focused derivation can be transformed into
a focused derivation.

In this paper, we present a focused sequent calculus for propositional
intuitionistic logic and prove the focalization property relative to a
standard presentation of propositional intuitionistic logic. Compared
to existing approaches, the proof is quite concise, depending only on
the internal soundness and completeness of the focused
logic. In turn, both of these properties can be established
(and mechanically verified) by structural induction in the style of
Pfenning's structural cut elimination without the need for any tedious
and repetitious invertibility lemmas. The proof of cut admissibility
for the focused system, which establishes internal soundness, is not
particularly novel. The proof of identity expansion, which establishes
internal completeness, is a major contribution of this work.
\end{abstract}

            
            

\begin{document}


\maketitle

\section{Introduction}

The propositions of intuitionistic propositional logic are easily recognizable
and standard: we will consider a logic with atomic propositions, 
falsehood, disjunction, truth, conjunction, and implication.
\[
\begin{array}{rcl}
P, Q & ::=
 & p \mid \bot \mid P_1 \vee P_2 \mid \top \mid P_1 \wedge P_2 
  \mid P_1 \supset P_2
\end{array}
\]
The sequent calculus presentation for intuitionistic logic is also standard;
the system in 
Figure~\ref{fig:unfoc} is precisely the propositional fragment of 
Kleene's sequent system $G_3$ as presented in \cite{pfenning00structural}.
Contexts $\Gamma$ are, as usual, considered to be
unordered multisets of propositions $P$, and the structural properties
of exchange, weakening, and contraction are admissible (each left
rule incorporates a contraction).

\begin{figure}
\fbox{$\seq{\Gamma}{P}$}
\[
\infer[\mathit{init}]
{\seq{\Gamma, p}{p} \mathstrut}{}
\qquad
\mbox{\it (no rule $\bot_R$)}
\qquad
\infer[\bot_L]
{\seq{\Gamma, \bot}{Q}}
{}
\]
\[
\infer[\vee_{R1}]
{\seq{\Gamma}{P_1 \vee P_2}}
{\seq{\Gamma}{P_1}}
\qquad
\infer[\vee_{R2}]
{\seq{\Gamma}{P_1 \vee P_2}}
{\seq{\Gamma}{P_2}}
\]
\[
\infer[\vee_L]
{\seq{\Gamma, P_1 \vee P_2}{Q}}
{\seq{\Gamma, P_1 \vee P_2, P_1}{Q}
 &
 \seq{\Gamma, P_1 \vee P_2, P_2}{Q}}
\]
\[
\infer[\top_R]
{\seq{\Gamma}{\top}}
{}
\qquad
\mbox{\it (no rule $\top_L$)}
\qquad
\infer[\wedge_R]
{\seq{\Gamma}{P_1 \wedge P_2}}
{\seq{\Gamma}{P_1}
 &
 \seq{\Gamma}{P_2}}
\]
\[
\infer[\wedge_{L1}]
{\seq{\Gamma, P_1 \wedge P_2}{Q}}
{\seq{\Gamma, P_1 \wedge P_2, P_1}{Q}}
\quad
\infer[\wedge_{L2}]
{\seq{\Gamma, P_1 \wedge P_2}{Q}}
{\seq{\Gamma, P_1 \wedge P_2, P_2}{Q}}
\]
\[
\infer[\supset_R]
{\seq{\Gamma}{P_1 \supset P_2}}
{\seq{\Gamma, P_1}{P_2}}
\qquad
\infer[\supset_L]
{\seq{\Gamma, P_1 \supset P_2}{Q}}
{\seq{\Gamma, P_1 \supset P_2}{P_1}
 &
 \seq{\Gamma, P_1 \supset P_2, P_2}{Q}}
\]
\caption{Sequent calculus for intuitionistic logic.}
\label{fig:unfoc}
\end{figure}

Sequent calculi are a nice way of presenting logics, and a logic's
sequent calculus presentation is a convenient setting in which to
establish the logic's metatheory in a way that is straightforwardly
mechanizable in proof assistants (like Twelf or Agda) that are
organized around the idea of structural induction.  There are two key
metatheoretic properties that we are interested in.  The first, cut
admissibility, justifies the use of lemmas: if we know $P$ (if we have
a derivation of the sequent $\seq{\Gamma}{P}$) and we know that $Q$
follows from assuming $P$ (if we have a derivation of the sequent
$\seq{\Gamma, P}{Q}$), then we can come to know $Q$ without the
additional assumption of $P$ (we can obtain a derivation of the
sequent $\seq{\Gamma}{Q}$). \footnote{In common practice, the words
 {\it proof} and {\it derivation} are used interchangeably.  In this
 article, we will be careful to refer to the formal objects
 constructed using sequent calculus rules (such as those in
 Figure~\ref{fig:unfoc}) as {\it derivations}.  Except when
 discussing natural deduction, the words {\it proof} and {\it
   theorem} will refer to theorems proved about these formal objects;
 these are frequently called {\it metatheorems} in the literature.} 
A
proof of the cut admissibility property establishes the {\it internal
  soundness} of a logic -- it implies that there are no closed 
derivations
of contradiction, even by circuitous reasoning using lemmas. The
identity property asserts that assuming $P$ is always sufficient to
conclude $P$, that is, that the sequent $\seq{\Gamma, P}{P}$ is always
derivable. A proof of the identity property establishes the {\it
  internal completeness} of a logic. We call these properties
\emph{internal}, following Pfenning \shortcite{pfenning10categorical},
to emphasize that these are properties of the deductive system itself
and not a comment on the system's relationship to any external
semantics.

There is a tradition in logic, dating back to Gentzen
\shortcite{gentzen35untersuchungen}, that views the sequent calculus
as a convenient formalism for proving a logic's metatheoretic
properties while viewing natural deduction proofs as the ``true proof
objects.''\footnote{This discussion assumes a basic familiarity with
  natural deduction. We refer the interested reader to
  Girard, Taylor, and Lafont's {\it Proofs and Types}
  \cite{girard89proofs}; the aforementioned quote comes from Section
  5.4 of that work.}  One reason for this bias towards natural
deduction is that natural deduction proofs have nice normalization
properties.  A natural deduction proof is {\it normal} if if there are
no instances of an introduction rule immediately followed by an
elimination rule of the same connective; such detours give rise to
{\it local reductions} which eliminate the detour, such as this one:
\[
\infer[\wedge_{E1}]
{P_1\,\mathit{true}}
{\infer[\wedge_I]
 {P_1 \wedge P_2\,\mathit{true}}
 {\deduce{P_1\,\mathit{true}}{\mathcal D_1}
  &
  \deduce{P_2\,\mathit{true}}{\mathcal D_2}}}
\qquad
\Longrightarrow_R
\qquad 
\deduce{P_1\,\mathit{true}}{\mathcal D_1}
\]
The {\it normalization} property says that every 
natural deduction proof can be transformed
into a normal natural deduction proof. 

We are frequently interested in the set of normal natural deduction
proofs of a given proposition.  As an example, there is exactly one
normal natural deduction proof for $(p \wedge q) \supset (r \wedge s)
\supset (p \wedge r)$. Presented as a derivation, that natural
deduction proof looks like this:
\[
\infer[\supset^u_I]
{(p \wedge q) \supset (r \wedge s) \supset (p \wedge r)\,\mathit{true}\mathstrut}
{\infer[\supset^v_I]
 {(r \wedge s) \supset (p \wedge r)\,\mathit{true}\mathstrut}
 {\infer[\wedge_I]
  {p \wedge r\,\mathit{true}\mathstrut}
  {\infer[\wedge_{E1}]
   {p\,\mathit{true}\mathstrut}
   {\infer[\it hyp_u]{p \wedge q\,\mathit{true}\mathstrut}{}}
   &
   \infer[\wedge_{E1}]
   {r\,\mathit{true}\mathstrut}
   {\infer[\it hyp_v]{r \wedge s\,\mathit{true}\mathstrut}{}}}}}
\]
Under the standard proof term assignment for natural deduction,
this (normal) natural deduction proof corresponds to the (irreducible) 
proof term $\lambda x. \lambda y. \langle \pi_1 x , \pi_1 y \rangle$.
 
In contrast, there are many sequent calculus derivations of the same
proposition.  Here's one of them:
\[
\infer[\supset_R]
{\seq{\cdot}{(p \wedge q) \supset (r \wedge s) \supset (p \wedge r)}}
{\infer[\supset_R]
 {\seq{p \wedge q}{(r \wedge s) \supset (p \wedge r)}}
 {\infer[\wedge_R]
  {\seq{p \wedge q, r \wedge s}{p \wedge r}}
  {\infer[\wedge_{L1}]
   {\seq{p \wedge q, r \wedge s}{p}}
   {\infer[\it init]
    {\seq{p \wedge q, r \wedge s, p}{p}}
    {}}
   &
   \infer[\wedge_{L1}]
   {\seq{p \wedge q, r \wedge s}{r}}
   {\infer[\it init]
    {\seq{p \wedge q, r \wedge s, r}{r}}
    {}}}}}
\]
Reading from bottom to top, this derivation decomposes 
$(p \wedge q) \supset (r \wedge s) \supset (p \wedge r)$ on the right, then
decomposes $(r \wedge s) \supset (p \wedge r)$ on the right, then decomposes
$p \wedge r$ on the right, and then (in one branch) decomposes $p \wedge q$
on the left while (in the other branch) decomposing $r \wedge s$ on the left.
Other possibilities include decomposing $p \wedge q$ on the left 
before decomposing $(r \wedge s) \supset (p \wedge r)$ on the right
and decomposing $r \wedge s$ on the left before $p \wedge r$ on the right;
there are at least six 
different derivations even if you don't count derivations
that do useless decompositions on the left. 

These different derivations are particularly problematic if our goal
is to do proof search for sequent calculus derivations, as inessential
differences between derivations correspond to unnecessary choice
points that a proof search procedure will need to backtrack over.  It
was in this context that Andreoli originally introduced the idea of
focusing. Some connectives, such as implication $A \supset B$, are
called {\it asynchronous} because their right rules can always be
applied eagerly, without backtracking, during bottom-up proof
search. Other connectives, such as disjunction $A \vee B$, are called
{\it synchronous} because their right rules cannot be applied
eagerly. For instance, the $\vee_{R1}$ rule cannot be applied eagerly
if we are looking for a derivation of $\seq{p}{\bot \vee p}$. 
This 
asynchronous or synchronous character is the connective's {\it
  polarity}.\footnote{Andreoli dealt with a one-sided classical sequent
  calculus; in intuitionistic logic, it is common to call asynchronous
  connectives {\it right}-asynchronous and {\it
    left}-synchronous. Similarly, it is common to call synchronous
  connectives {\it right}-synchronous and {\it left}-asynchronous.

  Synchronicity or polarity, a property of connectives, is closely
  connected to (and sometimes conflated with) a property of rules
  called {\it invertibility}; a rule is invertible if the conclusion
  of the rule implies each of the premises. So $\supset_R$ is
  invertible ($\seq{\Gamma}{P_1 \supset P_2}$ implies $\seq{\Gamma,
    P_1}{P_2}$) but $\supset_L$ is not ($\seq{\Gamma, P_1 \supset
    P_2}{C}$ does not imply $\seq{\Gamma, P_1 \supset P_2}{P_1}$).
  Rules that can be applied eagerly need to be invertible, so
  asynchronous connectives have invertible right rules and synchronous
  connectives have invertible left rules. Therefore, another synonym
  for asynchronous is {\it right-invertible}, and another synonym for
  synchronous is {\it left-invertible}. This terminology would be
  misleading in our case, as both the left and right rules for
  conjunction in Figure~\ref{fig:unfoc} are invertible.}

Andreoli's key observation was that proof search only needs to
consider derivations that have two alternating phases.  In {\it
  inversion} phases, we eagerly apply (invertible) right rules to
asynchronous connectives and (invertible) left rules to synchronous
ones. When this is no longer possible, we begin a {\it focusing} phase
by putting a single remaining proposition {\it in focus}, repeatedly
decomposing it (and only it) by applying right rules to synchronous
connectives and left rules to asynchronous ones. Andreoli described
this restricted form of proof search as a regular proof search
procedure in a restricted sequent calculus; such sequent calculi, and
derivations in them, are called {\it focused} as opposed to {\it
  unfocused} \cite{andreoli92logic}.

In order to adopt such a proof search strategy, it is important to
know that the strategy is both sound (i.e., the proof search strategy
will only say ``the sequent has a derivation'' if that is the case)
and complete (i.e., the proof search strategy is capable of finding a
derivation if one exists). Soundness proofs for focusing are usually easy:
focused derivations are essentially a syntactic refinement of the
unfocused derivations.  Completeness, the nontrivial direction,
involves turning unfocused derivations into focused ones. This process
is {\it focalization}.\footnote{The usage of {\it focus}, {\it
    focusing}, {\it focussing}, and {\it focalization} is not standard
  in the literature. We use the words {\it
    focus} and {\it focusing} to describe a logic (e.g. the focused
  sequent calculus) and aspects of that logic (e.g. focused
  derivations, propositions in focus, left- or right-focused sequents,
  and focusing phases).  {\it Focalization}, derived from the French
  {\it focalisation}, is reserved exclusively for the act of producing
  a focused derivation given an unfocused derivation; the focalization
  {\it property} establishes that focalization is always possible.} Thus, an effective procedure for
focalization is a constructive witness to the completeness of
focusing.

The techniques described in this article are general and can be
straightforwardly transferred to other modal and substructural logics,
as explored in the author's dissertation
\cite{simmons12substructural}. Our approach has three key
components, which we will now discuss in turn.

\paragraph*{Focalization via cut and identity}
Existing focalization proofs almost all fall prey to the need to prove
multiple tedious invertibility lemmas describing the interaction of
each rule with every other rule; this results in proofs that are
unrealistic to write out, difficult to check, and exhausting to
contemplate mechanizing. The way forward was first suggested by
Chaudhuri \shortcite{chaudhuri06focused}. In his dissertation, he
established the focalization property for linear logic as the
consequence of the focused logic's internal soundness (the cut
admissibility property) and completeness (the identity property).
Stating and proving the identity property for a focused sequent
calculus has remained a challenge, however. A primary contribution of
this work is {\it identity expansion}, a generalization of the
identity property that is amenable to mechanized proof by structural
induction on propositions. This identity property is, in turn, part of
our larger development, a proof of the focalization property that
entirely avoids the tedious invertibility lemmas that plague existing
approaches.  (We review existing techniques used to prove the
focalization property in Section~\ref{sec:previous}.)

\paragraph*{Refining the focused calculus}
The focused logic presented in this article is essentially equivalent
to the presentation of LJF given by Liang and Miller
\shortcite{liang09focusing}, a point we will return to in
Section~\ref{sec:logic-scalc}. A reader familiar with LJF will note
three non-cosmetic differences.  The first two, our use of a {\it
  polarized} variant of intuitionistic logic and our novel treatment
of atomic and suspended propositions, will be discussed further in
Section~\ref{sec:logic}.  A third change is that LJF does not force
any particular ordering for the application of rules during a
inversion phase. A seemingly inevitable consequence of this choice is
that the proof of focalization must establish the equivalence of all
permutations of these invertible rules; this is one of the
aforementioned tedious invertibility lemmas that plague proofs of the
focalization property. Our focused logic, like many others (including
Andreoli's original system) fixes a particular inversion order.

We introduce a new calculus rather than reusing an existing one in
order to present a focused logic and focalization proof that is
computationally clean and straightforward to both mechanize and
apply to other logics.  Our desire to
mechanize proofs of the focalization property also informed our
decision to use propositional intuitionistic logic.
All proofs in this
article are mechanized in both Twelf \cite{pfenning99system} and Agda
\cite{norell08towards}, though we will only mention the Twelf
development.
%
The concrete
basis for our claim of computational cleanliness is that our
mechanizations are complete artifacts capturing the constructive
content of the proofs we present, 
and the size of this artifact scales linearly
relative to the number of connectives; the approaches we call
``tedious'' tend to scale quadratically.

\paragraph*{Proof terms}
Since Andreoli's original work, focused sequent calculus derivations
have been shown to be isomorphic to normal natural deduction proofs
for restricted fragments of logic \cite{cervesato03linear} and
variations on the usual focusing discipline \cite{howe01proof}.  Such
results challenge the position that natural deduction proofs are
somehow more fundamental than sequent calculus derivations and also
indicate that focalization is a fundamental property of logic. In
Section~\ref{sec:proofterms}, we present a proof term language for
polarized intuitionistic logic that directly captures the branching
and binding structure of focused derivations. The result is a term
language generalizing the {\it spine form} of Cervesato and Pfenning
\shortcite{cervesato03linear}.

Understanding focalization at the level of proof terms is not strictly
necessary; the theorems we prove are perfectly sensible as statements
about sequent calculi.  We choose to present cut admissibility and
identity expansion at the level of proof terms in part because it
emphasizes the constructive content of those theorems. The
constructive content of cut admissibility is a substitution function
on proof terms generalizing the {\it hereditary substitution} of
Watkins et al.~\shortcite{watkins02concurrent} in a spine form
setting, and the constructive content of our identity expansion proof
is a novel $\eta$-expansion property on proof terms.

\begin{figure}
\begin{center}
\begin{tikzpicture}
\draw (6,20.5) node{``Soundness''};
\draw (10,20.5) node{``Completeness''};

\draw[dotted] (2,20) -- (12,20);
\draw (3,18.2) node{Focused,};
\draw (3,17.8) node{polarized};
\draw[dotted] (2,16) -- (12,16);
\draw (3,15) node{$\updownarrow$};
\draw[dotted] (2,14) -- (12,14);
\draw (3,13.2) node{Unfocused,};
\draw (3,12.8) node{unpolarized};

\draw (6,19.45) node{\bf Cut admissibility};
\draw (6,18.95) node{\it Theorem~\ref{thm:cut},};
\draw (6,18.55) node{\it Section~\ref{sec:cut}}; 

\draw (10,19.45) node{\bf Identity expansion};
\draw (10,18.95) node{\it Theorem~\ref{thm:identity},};
\draw (10,18.55) node{\it Section~\ref{sec:expansion}}; 

\draw (10,17.45) node{\bf Unfocused};
\draw (10,17.05) node{\bf admissibility lemmas};
\draw (10,16.55) node{\it Section~\ref{sec:unfocusedadmissibility}}; 

\draw (6,15.45) node{\bf De-focalization};
\draw (6,14.95) node{\it Theorem~\ref{thm:soundness},};
\draw (6,14.55) node{\it Section~\ref{sec:soundness}}; 

\draw (10,15.45) node{\bf Focalization};
\draw (10,14.95) node{\it Theorem~\ref{thm:completeness},};
\draw (10,14.55) node{\it Section~\ref{sec:focalizationproof}}; 

\draw (6,13.45) node{\bf Cut admissibility};
\draw (6,12.95) node{\it Corollary~\ref{cor:cut},};
\draw (6,12.55) node{\it Section~\ref{ref:corollaries}}; 

\draw (10,13.45) node{\bf Identity};
\draw (10,12.95) node{\it Corollary~\ref{cor:identity},};
\draw (10,12.55) node{\it Section~\ref{ref:corollaries}}; 

\draw [->] (4.6,19.3) .. controls (4,17) and (3.5,14) .. (4.4,13.6);
\draw [->] (11.6,19.3) .. controls (12.7,17) and (13,14) .. (10.8,13.6);

\draw [->] (4.6,19.3) .. controls (4,17) and (6.7,17) .. (8,17);
\draw [->] (11.6,19.3) .. controls (12.65,17) and (12,17.05) .. (11.95,17.05);

\draw [->] (10,16.3) -- (10,15.7);
\draw [->] (10,14.3) -- (10,13.7);
\draw [->] (6,14.3) -- (6,13.7);
\draw [->] (10,14.3) -- (6.3,13.72);
\draw [->] (6,14.3) -- (9.7,13.72);
\end{tikzpicture}
\end{center}
\caption{Theorems and their dependencies.}
\label{fig:outline}
\end{figure}
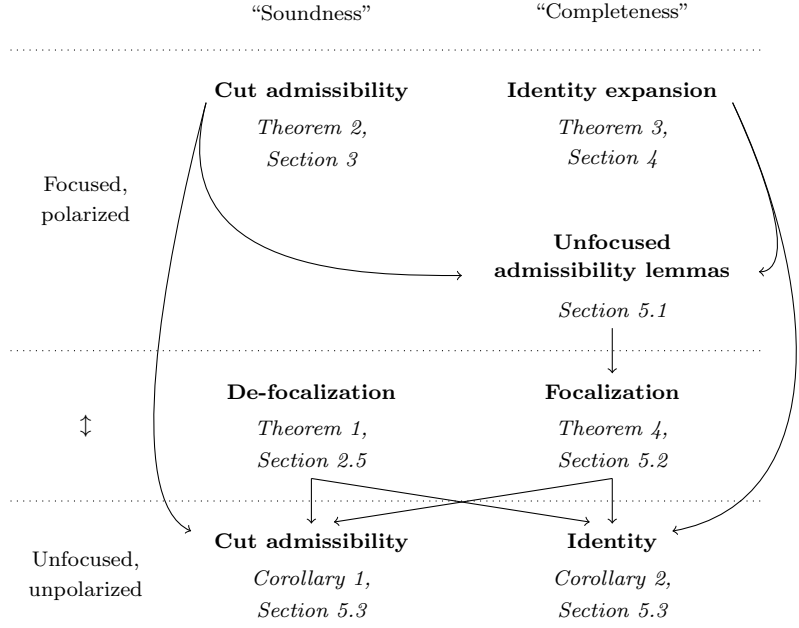

\subsection{Outline}

This article is dealing with three ``soundness'' properties and three
``completeness'' properties, so it is important to carefully 
explain what we're doing and when; the following discussion is represented
graphically in Figure~\ref{fig:outline}.

We present a new proof of the completeness {\it of focusing} (the
focalization property, Theorem~\ref{thm:completeness}) for
intuitionistic logic; the proof of the focalization property follows
from the {\it internal} soundness and completeness of the focused
sequent calculus (cut admissibility, Theorem~\ref{thm:cut}, and
identity expansion, Theorem~\ref{thm:identity}).  We will start, in
Section~\ref{sec:logic}, by motivating a polarized presentation of
logic that syntactically differentiates the synchronous and
asynchronous propositions. We then present a focused sequent calculus
for polarized propositional intuitionistic logic and formally state
the soundness and completeness of focusing. We also prove the
soundness of focusing (the de-focalization property,
Theorem~\ref{thm:soundness}) in this section, but it's pretty boring
and independent of the proofs of cut admissibility, identity
expansion, and the completeness of focusing.

Internal soundness for the focused sequent calculus is established by
the cut admissibility theorem in Section~\ref{sec:cut}, and internal
completeness for the focused sequent calculus is established in
Section~\ref{sec:expansion} using a generalization of the identity
expansion theorem first developed in \cite{simmons11weak}. In
Section~\ref{sec:focalization} we prove the focalization property by
showing {\it unfocused admissibility}, a group of lemmas establishing
that the focused sequent calculus can act like an unfocused sequent
calculus.  Finally, rather than proving the internal soundness and
completeness (cut and identity) for the unfocused system directly, we
show that these properties can be established as corollaries of the
first four theorems.  In Section~\ref{sec:previous} we conclude with
an overview of existing proofs of the focalization property.

We will henceforth avoid using 
the words {\it soundness} and {\it completeness} as much as possible. 
Instead, we will refer
to the cut admissibility and identity theorems 
for the focused and unfocused sequent calculi by name, and will refer
to the soundness and completeness of focusing as de-focalization 
and focalization, respectively.

\section{Polarized logic}
\label{sec:logic}

There is a significant line of work on polarity in logic dating back
to Andreoli \shortcite{andreoli92logic} and Girard
\shortcite{girard93unity}. That line of work holds that 
the asynchronous and
synchronous propositions are syntactic refinements of the set of
propositions. We can determine the synchronous or asynchronous
character of a proposition by inspecting its outermost
connective.\footnote{Linear logic naturally has two polarities. Other
  systems, like Girard's LU and Liang and Miller's LKU, use more than
  these two polarities
  \cite{girard93unity,liang11focused}. Furthermore, in LU the polarity
  of a proposition is determined by more than just the outermost
  connective.}

In a 1991 note published to the LINEAR list \cite{girard91sex}, Girard
introduced the idea of syntactically differentiating the {\it
  positive} propositions (those Andreoli called synchronous) from the
{\it negative} propositions (those Andreoli called asynchronous) while
mediating between the two with {\it shifts}: the upshift
${\uparrow}A^+$ includes positive propositions in the negative ones,
and the downshift ${\downarrow}A^-$ includes negative propositions in
the positive ones. This {\it polarization}\footnote{For the purposes
  of this article we are making somewhat artificial distinction
  between {\it polarity}, Andreoli's classification of propositions as
  asynchronous and synchronous, and {\it polarization}, the
  segregation of these two classes as positive or negative
  propositions by using shifts. This distinction is not a standard
  one: ``synchronous'' and ``positive'' are elsewhere used
  interchangeably, ``polarity'' is used to describe what we call
  ``polarization,'' and so on.} of logic was developed further by
Girard in Ludics \cite{girard01locus} and treated extensively in
Laurent and Zeilberger's Ph.D. theses
\cite{laurent02etude,zeilberger09logical}.
These are the propositions, positive and negative, for polarized
intuitionistic logic:
\[
\begin{array}{rcl}
A^+, B^+, C^+ & ::= & p^+ \mid {\downarrow} A^- \mid \bot \mid A^+ \vee B^+
  \mid \top^+ \mid A^+ \wedge^+ B^+
\\
A^-, B^-, C^- & ::= & p^- \mid {\uparrow} A^+ \mid A^+ \supset B^-
  \mid \top^- \mid A^- \wedge^- B^-
\end{array}
\]

Linear logic is able to unambiguously assign all connectives
to one category or the other, but in intuitionistic logic, truth
$\top$ and conjunction $P_1 \wedge P_2$
can be understood as having either a positive character (corresponding to
${\bf 1}$ and $A^+ \otimes B^+$ in linear logic) or a negative character
(corresponding to $\top$ and $A^- \with B^-$ in linear logic). 
We take the maximally general approach and allow both versions of 
truth and conjunction, which are decorated to emphasize their polarity.

\begin{figure}
\begin{align*}
(p^+)^\bullet & = p^+
  & (p^-)^\bullet & = p^-\\
({\downarrow}A^-)^\bullet & = (A^-)^\bullet
 & ({\uparrow}A^+)^\bullet & = (A^+)^\bullet\\
(\bot)^\bullet & = \bot\\
(A^+ \vee B^+)^\bullet & = (A^+)^\bullet \vee (B^+)^\bullet
 & (A^+ \supset B^-)^\bullet & = (A^+)^\bullet \supset (B^-)^\bullet\\
(\top^+)^\bullet & = \top
 & (\top^-)^\bullet & = \top\\
(A^+ \wedge^+ B^+)^\bullet & = (A^+)^\bullet \wedge (B^+)^\bullet
 & (A^- \wedge^- B^-)^\bullet & = (A^-)^\bullet \wedge (B^-)^\bullet
\end{align*}
\caption{Erasure of polarized propositions.}
\label{fig:erasure}
\end{figure}

The shifts introduced by Girard were modalities that might change the
provability of a proposition. We adopt the later stance of Zeilberger
\shortcite{zeilberger09logical}, McLaughlin and Pfenning
\shortcite{mclaughlin09efficient}, and others: shifts influence the
structure of derivations, but not the provability of
propositions. Therefore, we expect there to be a focused derivation of
$A^+$ or $A^-$ if and only if there is an unfocused derivation of
$(A^+)^\bullet$ or $(A^-)^\bullet$, where $(-)^\bullet$ is the {\it
  erasure} function given in Figure~\ref{fig:erasure}.  The polarity
of an atomic proposition can be arbitrary as long as it is consistent,
as if each individual atomic proposition $p$ is really intrinsically
positive or negative, but the calculus in Figure~\ref{fig:unfoc}
didn't notice.

Shifts and polarization are computationally interesting phenomena. Our
view of polarization lines up with the {\it call-by-push-value} system
independently developed by Levy: positive propositions correspond to
{\it value types} and negative propositions correspond to {\it
  computation types} \cite{levy04call}.  Shifts are also useful in
theorem proving. By employing different {\it polarization strategies},
the name for partial inverses of erasure, the Imogen theorem prover can
simulate fully-focused LJF proof search, proof search in an unfocused
logic like Kleene's $G_3$, and proof search in many other
partially-focused systems in between \cite{mclaughlin09efficient}.
Furthermore, forward-chaining versus backward-chaining logic
programming can be seen as arising from particular polarization
strategies in uniform proof search \cite{chaudhuri10logical}. However,
in this article polarization is primarily a technical device. Shifts
makes it clearer that certain proofs are structurally inductive over
propositions, especially identity expansion
(Theorem~\ref{thm:identity}).  Additionally, the focus-interrupting
nature of the shift is a critical part of our unfocused admissibility
lemmas.\footnote{Something like a shift is necessary for Liang and
  Miller's focalization proofs as well. Because their logic has
  polarity but no shifts, they construct {\it delays} $\delta^+(A) =
  \top^+ \wedge^+ A$ and $\delta^-(A) = \top^+ \supset A$ to force
  asynchronous or synchronous connectives to behave (respectively)
  like synchronous or asynchronous ones \cite{liang09focusing}. This
  gets in the way of defining logical fragments by adding or removing
  connectives, as the presence of the connective $\top^+$ is
  fundamental to their completeness proof.}

\subsection{Sequent calculus}\label{sec:logic-scalc}

We will develop our focused sequent calculus in two stages; in the
first stage we do not consider atomic propositions. We can present
sequents for our polarized logic in two equivalent ways. In the {\it
  one-sequent} view, we say that all sequents have the form
$\efoc{\Gamma}{L}{U}$. The components of a sequent are defined by the
following (currently incomplete) grammar:
\begin{align*}
& \mbox{\it Hypothetical contexts} & 
\Gamma & ::= \cdot \mid \Gamma, A^- \mid \ldots
\\
& \mbox{\it Inversion contexts} &
\Omega & ::= \cdot \mid A^+, \Omega
\\
& \mbox{\it Antecedents} & 
L & ::= \Omega \mid [A^-]
\\
& \mbox{\it Succedents} & 
U & ::= [A^+] \mid A^+ \mid A^- \mid \ldots
\end{align*}
This first view requires us to further restrict the form of sequents
for two reasons. First, we only want to focus on one proposition at a
time, so only one right focus $[A^+]$ or left focus $[A^-]$ should be
present in a sequent.  Second, focus and inversion phases should not
overlap, so it must be the case that $L = \cdot$ when $U = [A^+]$ and
that $U \neq A^-$ when $L = [A^-]$. Given this
restriction, the first view of sequents is equivalent to a {\it
  three-sequent} view in which there are three different sequent
forms:
\begin{itemize}
\item Right focus: $\rfoc{\Gamma}{A^+}$, where $L = \cdot$ and is
  therefore omitted,
\item Inversion: $\ifoc{\Gamma}{\Omega}{U}$, where $U \neq [A^+]$, and
\item Left focus: $\lfoc{\Gamma}{A^-}{U}$, where $U$ is {\it stable}
  (more about this shortly).
\end{itemize}

\begin{figure}
\fbox{$\rfoc{\Gamma}{A^+}$ -- {\it right focus}} 
\[
\infer[{\downarrow}_R]
{\rfoc{\Gamma}{{\downarrow}A^-}}
{\ifoc{\Gamma}{\cdot}{A^-}}
\]
\[
\mbox{\it (no rule $\bot_R$)}
\qquad
\infer[\vee_{R1}]
{\rfoc{\Gamma}{A^+ \vee B^+}}
{\rfoc{\Gamma}{A^+}}
\qquad
\infer[\vee_{R2}]
{\rfoc{\Gamma}{A^+ \vee B^+}}
{\rfoc{\Gamma}{B^+}}
\]
\[
\infer[\top^+_R]
{\rfoc{\Gamma}{\top^+}}
{}
\qquad
\infer[\wedge^+_R]
{\rfoc{\Gamma}{A^+ \wedge^+ B^+}}
{\rfoc{\Gamma}{A^+}
 &
 \rfoc{\Gamma}{B^+}}
\]

\medskip
\fbox{$\ifoc{\Gamma}{\Omega}{U}$ -- {\it inversion, $U \neq [A^+]$}}
\[
\infer[{\it foc}_R]
{\ifoc{\Gamma}{\cdot}{A^+}}
{\rfoc{\Gamma}{A^+}}
\qquad
\infer[{\it foc}_L]
{\ifoc{\Gamma, A^-}{\cdot}{U}}
{\stable{U}
 &
 \lfoc{\Gamma, A^-}{A^-}{U}}
\]
\[
\infer[\downarrow_L]
{\ifoc{\Gamma}{{\downarrow}A^-, \Omega}{U}}
{\ifoc{\Gamma, A^-}{\Omega}{U}}
\]
\[
\infer[\bot_L]
{\ifoc{\Gamma}{\bot, \Omega}{U}}
{}
\qquad
\infer[\vee_L]
{\ifoc{\Gamma}{A^+ \vee B^+, \Omega}{U}}
{\ifoc{\Gamma}{A^+, \Omega}{U}
 &
 \ifoc{\Gamma}{B^+, \Omega}{U}}
\]
\[
\infer[\top^+_L]
{\ifoc{\Gamma}{\top^+, \Omega}{U}}
{\ifoc{\Gamma}{\Omega}{U}}
\qquad
\infer[\wedge^+_L]
{\ifoc{\Gamma}{A^+ \wedge^+ B^+, \Omega}{U}}
{\ifoc{\Gamma}{A^+, B^+, \Omega}{U}}
\]
\[
\infer[{\uparrow}_R]
{\ifoc{\Gamma}{\cdot}{{\uparrow}A^+}}
{\ifoc{\Gamma}{\cdot}{A^+}}
\qquad
\infer[\supset_R]
{\ifoc{\Gamma}{\cdot}{A^+ \supset B^-}}
{\ifoc{\Gamma}{A^+}{B^-}}
\]
\[
\infer[\top^-_R]
{\ifoc{\Gamma}{\cdot}{\top^-}}
{}
\qquad
\infer[\wedge^-_R]
{\ifoc{\Gamma}{\cdot}{A^- \wedge^- B^-}}
{\ifoc{\Gamma}{\cdot}{A^-}
 &
 \ifoc{\Gamma}{\cdot}{B^-}}
\]

\medskip
\fbox{$\lfoc{\Gamma}{A^-}{U}$ -- {\it left focus, $U$ must be stable}}
\[
\infer[{\uparrow}_L]
{\lfoc{\Gamma}{{\uparrow}A^+}{U}}
{\ifoc{\Gamma}{A^+}{U}}
\qquad
\infer[{\supset}_L]
{\lfoc{\Gamma}{A^+ \supset B^-}{U}}
{\rfoc{\Gamma}{A^+}
 &
 \lfoc{\Gamma}{B^-}{U}}
\]
\[
\mbox{\it (no rule $\top^-_L$)}
\qquad
\infer[\wedge^-_{L1}]
{\lfoc{\Gamma}{A^- \wedge^- B^-}{U}}
{\lfoc{\Gamma}{A^-}{U}}
\qquad
\infer[\wedge^-_{L2}]
{\lfoc{\Gamma}{A^- \wedge^- B^-}{U}}
{\lfoc{\Gamma}{B^-}{U}}
\]

\medskip
\fbox{$\stable{U}$}
\vspace{-12pt}
\[
\infer
{\stable{A^+}}
{}
\]
\caption{Focused sequent calculus for polarized intuitionistic logic
  (sans suspended propositions).}
\label{fig:foc}
\end{figure}

The sequent calculus for polarized intuitionistic logic in
Figure~\ref{fig:foc} is presented in terms of this three-sequent view.
The {\it right focus} sequent $\rfoc{\Gamma}{A^+}$ describes a state
in which non-invertible right rules are being applied to positive
propositions, the {\it left focus} sequent $\lfoc{\Gamma}{A^-}{U}$
describes a state in which non-invertible left rules are being applied
to negative propositions, and the {\it inversion} sequent
$\ifoc{\Gamma}{\Omega}{U}$ describes everything else. While we treat
the hypothetical context $\Gamma$ informally as a multiset, the {\it
  inversion context} $\Omega$ is {\it not} a multiset. Instead, it
should be thought of as an ordered sequence of positive propositions:
the empty sequence is written as ``$\cdot$'' and ``$,$'' is an
associative append operator.  Whenever the inversion context $\Omega$ is
non-empty, the {\it only} applicable rule is the one that decomposes
the left-most positive connective in $\Omega$.

The picture in Figure~\ref{fig:foc} is quite uniform: every rule
except for ${\it foc}_R$ and ${\it foc}_L$ breaks down a single
connective, while the shift rules regulate the focus and inversion
phases.  Read from bottom up, they put an end to the process of
breaking down a proposition under focus (${\downarrow}_R$,
${\uparrow}_L$) and to the process of breaking down a proposition with
inversion (${\downarrow}_L$, ${\uparrow}_R$). All other rules maintain
focus or inversion on the subformulas of a proposition.

The conclusions of ${\it foc}_R$ and ${\it foc}_L$ are inversion
sequents with empty inversion contexts and succedents $U$ that are
stable, meaning that there is no possibility of applying an invertible
rule.  These {\it stable sequents} (sometimes called {\it neutral})
have an important place in focused sequent calculi.  In the
introduction, we claimed that there was an essential equivalence
between LJF and our focused presentation, but this is only true if we
ignore the internal structure of focus and inversion phases and work
with {\it synthetic rules}, the derivation fragments comprised of one
or more inversion phases stacked on top of a single focused phase. In
the synthetic view of focusing, we abstract away from the internal
structure of focusing phases to emphasize the stable sequents that lie
between them \cite{andreoli01focussing}. LJF and our presentation of
polarized intuitionistic logic have different internal structure, but
we claim that the systems give rise to the same synthetic
rules.\footnote{Our presentation does not have the first-order
  quantifiers present in LJF, and LJF lacks a negative unit $\top^-$,
  but quantifiers can be added to our system easily, and the same is
  true for $\top^-$ in LJF.} Aside from our use of shifts, we depart
from Liang and Miller in ways that are technically relevant but that
are invisible at the level of synthetic connectives.

The stability requirement for $U$ in left focus sequents
$\lfoc{\Gamma}{A^-}{U}$ can equivalently be stated as a extra premise
$\stable{U}$ for the rule ${\uparrow}_L$, which is done in the
accompanying Twelf development. (Our placement of the premise on ${\it
  foc}_L$ shows a bit of bias towards bottom-up proof construction.)

\subsection{Suspended propositions}\label{sec:susp}

The pleasant picture of focusing given above must become more
complicated when we
consider atomic propositions. Atomic propositions are best understood
as stand-ins for arbitrary propositions, and so our polarized logic
has both positive atomic propositions (stand-ins for arbitrary
positive propositions) and negative atomic propositions (stand-ins for
arbitrary negative propositions).

When we are performing inversion and we reach an atomic proposition,
we do not have enough information to break down that proposition any
further, but we have not reached a shift. We have to do something
different. What we do is {\it suspend} that atomic proposition, either
in the hypothetical context or in the succedent. We represent a
suspended atomic proposition as $\susp{p^+}$ or $\susp{p^-}$. If we
wanted to closely follow existing focused sequent calculi, we would
introduce two more rules for proving atomic propositions in focus
using a suspended atomic proposition. The resulting extension of
Figure~\ref{fig:foc} would look something like this:
\[
\infer
{\ifoc{\Gamma}{p^+, \Omega}{U}}
{\ifoc{\Gamma, \susp{p^+}}{\Omega}{U}}
\qquad
\infer
{\ifoc{\Gamma}{\cdot}{p^-}}
{\ifoc{\Gamma}{\cdot}{\susp{p^-}}}
\qquad
\infer
{\rfoc{\Gamma, \susp{p^+}}{p^+}}
{}
\qquad
\infer
{\lfoc{\Gamma}{p^-}{\susp{p^-}}}
{}
\qquad
\infer
{\stable{\susp{p^-}}}
{}
\]

This treatment is not incorrect and is obviously analogous to the ${\it
  init}$ rule from the unfocused system in
Figure~\ref{fig:unfoc}. Nevertheless, we contend that this is a design
error and a large part of why it has historically been difficult to
prove the identity theorem for focused systems. We generalize the
rules above
by allowing the hypothetical context to contain arbitrary
suspended positive propositions (not just atomic positive
propositions) and allowing the succedent to contain arbitrary
suspended negative propositions (not just atomic negative
propositions).

\begin{figure}
\[
\infer[{\it id}^+]
{\rfoc{\Gamma, \susp{A^+}}{A^+}}
{}
\qquad
\infer[{\it id}^-]
{\lfoc{\Gamma}{A^-}{\susp{A^-}}}
{}
\]
\[
\infer[\eta^+]
{\ifoc{\Gamma}{p^+, \Omega}{U}}
{\ifoc{\Gamma, \susp{p^+}}{\Omega}{U}}
\qquad
\infer[\eta^-]
{\ifoc{\Gamma}{\cdot}{p^-}}
{\ifoc{\Gamma}{\cdot}{\susp{p^-}}}
\qquad
\infer
{\stable{\susp{A^-}}}
{}
\]
\caption{Focused sequent calculus, extended with suspended propositions}
\label{fig:foc-atom}
\end{figure}

This generalization allows us to finally give the complete grammar of 
hypothetical contexts and succedents: 
\begin{align*}
& \mbox{\it Hypothetical contexts} & 
\Gamma & ::= \cdot \mid \Gamma, A^- \mid \Gamma, \susp{A^+}
\\
& \mbox{\it Succedents} & 
U & ::= [A^+] \mid A^+ \mid A^- \mid \susp{A^-}
\end{align*}
The rules for atomic propositions, extending Figure~\ref{fig:foc}, are
given in Figure~\ref{fig:foc-atom}. The $\eta^+$ and $\eta^-$ rules
are the same as the ones we discussed above, reflecting the fact that
the inversion process {\it must} suspend itself at an atomic
proposition and {\it should not} suspend itself any earlier.  The
${\it id}^+$ and ${\it id}^-$ rules directly describe an identity
or hypothesis principle, but only for suspended propositions.

These more general ${\it id}^+$ and ${\it id}^-$ rules allow us to define two
substitution principles, which are critical for the proof of identity
expansion in Section~\ref{sec:expansion}. The derivation of a
right-focused sequent $[A^+]$ can discharge $\susp{A^+}$ in the
hypothetical context, and the derivation of a left-focused sequent
$[A^-]$ can discharge $\susp{A^-}$ in the succedent. Written as
admissible rules, these two {\it focal substitution} principles are as
follows:
\[
\infer-[{\it subst}^+]
{\efoc{\Gamma}{L}{U}}
{\rfoc{\Gamma}{A^+}
 &
 \efoc{\Gamma, \susp{A^+}}{L}{U}}
\qquad
\infer-[{\it subst}^-]
{\efoc{\Gamma}{L}{U}}
{\efoc{\Gamma}{L}{\susp{A^-}}
 &
 \lfoc{\Gamma}{A^-}{U}}
\]
It is straightforward to establish the positive focal substitution
principle by induction over the derivation of $\efoc{\Gamma,
  \susp{A^+}}{L}{U}$, and it is likewise straightforward to establish
the negative substitution principle by induction over the derivation
of $\efoc{\Gamma}{L}{\susp{A^-}}$. When the last rule in the
derivation we're inducting over is ${\it id}^+$ or ${\it id}^-$, we
return the derivation we're not inducting over, and in every other
case we apply the induction hypothesis directly.\footnote{In the case
  of the rules $\eta^+$ and ${\downarrow}_L$, we also apply an
  admissible weakening principle to the derivation we're not inducting
  over.}

The admissible rules ${\it subst}^+$ and ${\it subst}^-$ are {\it
  uniform} substitution principles. This means that, in the
accompanying Twelf development, it is possible to get them both for
free from the LF function space, the same way we get weakening 
and contraction of the
hypothetical context for free and generally take it for granted. This
is natural in the case of positive focal substitution: we interpret
the suspended atomic proposition $\susp{A^+}$ as a uniform assumption
that $A^+$ is provable in right focus. It is more counterintuitive to
get negative focal substitution for free in LF; we refer the reader to
the accompanying Twelf development for details. 

The logic extended with these more general ${\it id}^+$ and ${\it
  id}^-$ rules conservatively extends the logic with the more
traditional rules we initially proposed. Reading rules from bottom to
top, the rules $\eta^+$ and $\eta^-$ are the only ones that introduce
suspended propositions. Therefore, given the derivation of a sequent
where every suspended proposition is atomic, we know that every
instance of ${\it id}^+$ and ${\it id}^-$ in that derivation acts on an
atomic proposition. We call sequents where every suspended proposition
is atomic {\it suspension-normal} sequents. Certain operations, in
particular erasure and cut admissibility, are only defined on
suspension-normal sequents and derivations of these sequents.

\subsection{Erasure and focalization}

\begin{figure}
\begin{align*}
 & (\Omega)^\bullet
 & (\cdot)^\bullet & = \cdot
 & (A^+, \Omega)^\bullet & = (A^+)^\bullet, (\Omega)^\bullet
\\
 & (\Gamma)^\circledast
 & (\cdot)^\circledast & = \cdot 
 & (\Gamma, A^-)^\circledast & = (\Gamma)^\circledast, (A^-)^\bullet
 & (\Gamma, \susp{p^+})^\circledast & = (\Gamma)^\circledast, p^+
\\
 & (L)^\circledast
 & ([A^-])^\circledast & = (A^-)^\bullet
 & (\Omega)^\circledast & = (\Omega)^\bullet
\\
 & (U)^\circledast
 & ([A^+])^\circledast & = (A^+)^\bullet
 & (A^+)^\circledast & = (A^+)^\bullet
 & (A^-)^\circledast & = (A^-)^\bullet
 & (\susp{p^-})^\circledast & = p^-
\end{align*}
\caption{Erasure of contexts and succedents. $(A^+)^\bullet$ and
  $(A^-)^\bullet$ are defined in Figure~\ref{fig:erasure}.}
\label{fig:erasure-ctx}
\end{figure}

We presented the erasure of propositions in Figure~\ref{fig:erasure},
and Figure~\ref{fig:erasure-ctx} describes the erasure of a polarized
contexts and sequents. Note that erasure is only defined on
hypothetical contexts $\Gamma$ and succedents $U$ that are
suspension-normal.
 
Erasure is a pretty boring operation, important mainly because it
allows us to state soundness and completeness of focusing. We want to
understand completeness in terms of stable, suspension-normal
sequents, so the correctness of focusing states that, if $\Gamma$
and $U$ are stable and suspension-normal, $\ifoc{\Gamma}{\cdot}{U}$ if
and only if $\seq{(\Gamma)^\circledast}{(U)^\circledast}$. The
backward (completeness) direction is focalization and the forward
(soundness) direction is de-focalization.

Many different polarized propositions will typically erase to the same
unpolarized proposition.  The proposition used in the example from the
introduction, $(p \wedge q) \supset (r \wedge s) \supset (p \wedge
r)$, is the erasure of each of the following:
\begin{align}
{\downarrow}(p \wedge^- q) \supset
 {\downarrow}(r \wedge^- s) & \supset 
 (p \wedge^- r)\label{eqn:neg}\\
(p \wedge^+ q) \supset 
 (r \wedge^+ s) & \supset
 {\uparrow}(p \wedge^+ r)\label{eqn:pos}\\
{\downarrow}({\uparrow}p \wedge^- {\uparrow}q) \supset
 {\uparrow}{\downarrow}({\downarrow}({\uparrow}r \wedge^- {\uparrow}s) & \supset
 {\uparrow}{\downarrow}({\uparrow}p \wedge^- {\uparrow}r))\label{eqn:seq}
\end{align}
Note that the first proposition implies a 
negative polarity for all atomic propositions and the last two propositions
imply a positive polarity. 

The first and second propositions
each have {\it exactly one} focused 
derivation, just as the unpolarized propositions
had exactly one normal natural deduction proof. The unique derivation of 
(\ref{eqn:neg}) is structurally similar to the unfocused derivation from the
introduction:
\[
\infer[\supset_R]
{\ifoc{\cdot}{\cdot}
      {{\downarrow}(p \wedge^- q) \supset
       {\downarrow}(r \wedge^- s) \supset 
       (p \wedge^- r)}}
{\infer[{\downarrow}_L]
 {\ifoc{\cdot}{{\downarrow}(p \wedge^- q)}{{\downarrow}(r \wedge^- s) \supset 
  (p \wedge^- r)}}
 {\infer[\supset_R]
  {\ifoc{p \wedge^- q}{\cdot}
        {{\downarrow}(r \wedge^- s) \supset (p \wedge^- r)}}
  {\infer[{\downarrow}_L]
   {\ifoc{p \wedge^- q}{{\downarrow}(r \wedge^- s)}{p \wedge^- r}}
   {\infer[\wedge^-_R]
    {\ifoc{p \wedge^- q, r \wedge^- s}{\cdot}{p \wedge^- r}}
    {\infer[\eta^-]
     {\ifoc{p \wedge^- q, r \wedge^- s}{\cdot}{p}}
     {\infer[{\it foc}_L]
      {\ifoc{p \wedge^- q, r \wedge^- s}{\cdot}{\susp{p}}}
      {\infer[\wedge^-_{L1}]
       {\lfoc{p \wedge^- q, r \wedge^- s}{p \wedge^- q}{\susp{p}}}
       {\infer[{\it id}^-]
        {\lfoc{p \wedge^- q, r \wedge^- s}{p}{\susp{p}}}
        {}}}}
     &
     \infer[\eta^-]
     {\ifoc{p \wedge^- q, r \wedge^- s}{\cdot}{r}}
     {\infer[{\it foc}_L]
      {\ifoc{p \wedge^- q, r \wedge^- s}{\cdot}{\susp{r}}}
      {\infer[\wedge^-_{L1}]
       {\lfoc{p \wedge^- q, r \wedge^- s}{r \wedge^- s}{\susp{r}}}
       {\infer[{\it id}^-]
        {\lfoc{p \wedge^- q, r \wedge^- s}{r}{\susp{r}}} 
        {}}}}}}}}}
\]

The unique derivation of (\ref{eqn:pos}) decomposes the proposition
in a completely different order:
\[
\infer[\supset_R]
{\ifoc{\cdot}{\cdot}
      {(p \wedge^+ q) \supset (r \wedge^+ s) \supset {\uparrow}(p \wedge^+ r)}}
{\infer[\wedge^+_L]
 {\ifoc{\cdot}{p \wedge^+ q}
       {(r \wedge^+ s) \supset {\uparrow}(p \wedge^+ r)}}
 {\infer[\eta^+]
  {\ifoc{\cdot}{p, q}{(r \wedge^+ s) \supset {\uparrow}(p \wedge^+ r)}}
  {\infer[\eta^+]
   {\ifoc{\susp{p}}{q}{(r \wedge^+ s) \supset {\uparrow}(p \wedge^+ r)}}
   {\infer[\supset_R]
    {\ifoc{\susp{p}, \susp{q}}{\cdot}
         {(r \wedge^+ s) \supset {\uparrow}(p \wedge^+ r)}}
    {\infer[\wedge^+_L]
     {\ifoc{\susp{p}, \susp{q}}{r \wedge^+ s}{{\uparrow}(p \wedge^+ r)}}
     {\infer[\eta^+]
      {\ifoc{\susp{p}, \susp{q}}{r, s}{{\uparrow}(p \wedge^+ r)}}
      {\infer[\eta^+]
       {\ifoc{\susp{p}, \susp{q}, \susp{r}}{s}{{\uparrow}(p \wedge^+ r)}}
       {\infer[{\uparrow}_R]
        {\ifoc{\susp{p}, \susp{q}, \susp{r}, \susp{s}}{\cdot}
              {{\uparrow}(p \wedge^+ r)}}
        {\infer[{\it foc}_R]
         {\ifoc{\susp{p}, \susp{q}, \susp{r}, \susp{s}}{\cdot}{p \wedge^+ r}}
         {\infer[\wedge^+_R]
          {\rfoc{\susp{p}, \susp{q}, \susp{r}, \susp{s}}{p \wedge^+ r}}
          {\infer[{\it id}^+]
           {\rfoc{\susp{p}, \susp{q}, \susp{r}, \susp{s}}{p}}
           {}
           &
           \infer[{\it id}^+]
           {\rfoc{\susp{p}, \susp{q}, \susp{r}, \susp{s}}{r}}
           {}}}}}}}}}}}}
\]
These examples illustrate how polarity and focusing can dramatically
reduce the bureaucratic nondeterminism present in the unfocused
sequent calculus. To be clear, however, we have {\it chosen} to reduce
that bureaucratic nondeterminism: the derivations of proposition
(\ref{eqn:seq}) are isomorphic to the unfocused derivations of the
original, unpolarized proposition.

Our statement of the focalization property applies to {\it all}
polarization strategies.  Having obtained this strong focalization
property, if we are given an unfocused derivation of any unpolarized
sequent (such as $\seq{\cdot}{(p \wedge q) \supset (r \wedge s)
  \supset (p \wedge r)}$), we can use any polarization strategy at our
disposal to turn the unpolarized sequent into a polarized sequent
(such as $\ifoc{\cdot}{\cdot}{{\downarrow}((p \wedge^+ q) \supset (r
  \wedge^+ s) \supset {\downarrow}(p \wedge^+ r))}$) and then use
focalization to transform the unfocused derivation into a focused
derivation. Our proof that unfocused cut and identity follow from
focused cut and identity does require that we know about {\it some}
polarization strategy, but that will be the extent to which the
present technical development relies on the matter.

With the exception of Zeilberger \shortcite{zeilberger08unity}, proofs
of the focalization property tend not to operate on the basis of
erasure. Erasure-based polarization only emerges clearly as an option
in a logic with shifts; Andreoli's focused classical linear logic
\shortcite{andreoli92logic}, Chaudhuri's focused intuitionistic linear
logic \shortcite{chaudhuri06focused}, and Liang and Miller's LJF
\shortcite{liang09focusing} all approach focalization for a logic
where there are no shifts and where polarity is derived from a
proposition's topmost connective.  From our polarized perspective,
these approaches can all be seen as defining a particular polarization
strategy $(\seq{\Gamma}{P})^\circ$ that transforms unpolarized
sequents into polarized ones. It is then possible to state and prove a
strictly weaker focalization property: that $(\seq{\Gamma}{P})^\circ$
is derivable if and only if $\seq{\Gamma}{P}$ is derivable.

\subsection{Proof terms}\label{sec:proofterms}

\begin{figure}
\fbox{$\rfoct{\Gamma}{V}{A^+}$ -- {\it values} $V$}
\[
\infer[{\it id}^+]
{\rfoct{\Gamma, z{:}\susp{A^+}}{z}{A^+}}
{}
\qquad
\infer[{\downarrow}_R]
{\rfoct{\Gamma}{\dsrt{N}}{{\downarrow}A^-}}
{\ifoct{\Gamma}{\cdot}{N}{A^-}}
\]
\[
\mbox{\it (no rule $\bot_R$)}
\qquad
\infer[\vee_{R1}]
{\rfoct{\Gamma}{\mathsf{inl}\,V}{A^+ \vee B^+}}
{\rfoct{\Gamma}{V}{A^+}}
\qquad
\infer[\vee_{R2}]
{\rfoct{\Gamma}{\mathsf{inr}\,V}{A^+ \vee B^+}}
{\rfoct{\Gamma}{V}{B^+}}
\]
\[
\infer[\top^+_R]
{\rfoct{\Gamma}{\langle\rangle^+}{\top^+}}
{}
\qquad
\infer[\wedge^+_R]
{\rfoct{\Gamma}{\langle V_1, V_2 \rangle^+}{A^+ \wedge^+ B^+}}
{\rfoct{\Gamma}{V_1}{A^+} 
 &
 \rfoct{\Gamma}{V_2}{B^+}}
\]

\medskip
\fbox{$\ifoct{\Gamma}{\Omega}{N}{U}$ -- {\it terms $N$, $U \neq [A^+]$}}
\[
\infer[{\it foc}_R]
{\ifoct{\Gamma}{\cdot}{\rft{V}}{A^+}}
{\rfoct{\Gamma}{V}{A^+}}
\qquad
\infer[{\it foc}_L]
{\ifoct{\Gamma, x{:}A^-}{\cdot}{\lft{x}{S}}{U}}
{\stable{U}
 &
 \lfoct{\Gamma, x{:}A^-}{S}{A^-}{U}}
\]
\[
\infer[\eta^+]
{\ifoct{\Gamma}{p^+, \Omega}{\etalt{z}{N}}{U}}
{\ifoct{\Gamma, z{:}\susp{p^+}}{\Omega}{N}{U}}
\qquad
\infer[{\downarrow}_L]
{\ifoct{\Gamma}{{\downarrow}A^-, \Omega}{\dslt{x}{N}}{U}}
{\ifoct{\Gamma, x{:}A^-}{\Omega}{N}{U}}
\]
\[
\infer[\bot_L]
{\ifoct{\Gamma}{\bot, \Omega}{\mathsf{abort}}{U}}
{}
{}
\qquad
\infer[\vee_L]
{\ifoct{\Gamma}{A^+ \vee B^+, \Omega}{[N_1, N_2]}{U}}
{\ifoct{\Gamma}{A^+, \Omega}{N_1}{U}
 &
 \ifoct{\Gamma}{B^+, \Omega}{N_2}{U}}
\]
\[
\infer[\top^+_L]
{\ifoct{\Gamma}{\top^+, \Omega}{\langle\rangle.N}{U}}
{\ifoct{\Gamma}{\Omega}{N}{U}}
\qquad
\infer[\wedge^+_L]
{\ifoct{\Gamma}{A^+ \wedge^+ B^+, \Omega}{\times N}{U}}
{\ifoct{\Gamma}{A^+, B^+, \Omega}{N}{U}}
\]
\[
\infer[\eta^-]
{\ifoct{\Gamma}{\cdot}{\etart{N}}{p^-}}
{\ifoct{\Gamma}{\cdot}{N}{\susp{p^-}}}
\qquad
\infer[{\uparrow}_R]
{\ifoct{\Gamma}{\cdot}{\usrt{N}}{{\uparrow}A^+}}
{\ifoct{\Gamma}{\cdot}{N}{A^+}}
\qquad
\infer[{\supset}_R]
{\ifoct{\Gamma}{\cdot}{\lambda N}{A^+ \supset B^-}}
{\ifoct{\Gamma}{A^+}{N}{B^-}}
\]
\[
\infer[{\top^-}_R]
{\ifoct{\Gamma}{\cdot}{\langle\rangle^-}{\top^-}}
{}
\qquad
\infer[{\wedge}^-_R]
{\ifoct{\Gamma}{\cdot}{\langle N_1, N_2 \rangle^-}{A^- \wedge^- B^-}}
{\ifoct{\Gamma}{\cdot}{N_1}{A^-}
 &
 \ifoct{\Gamma}{\cdot}{N_2}{B^-}}
\]

\medskip
\fbox{$\lfoct{\Gamma}{S}{A^-}{U}$ -- {\it spines $S$, $U$ must be stable}}
\[
\infer[{\it id}^-]
{\lfoct{\Gamma}{\textsc{nil}}{A^-}{\susp{A^-}}}
{}
\qquad
\infer[{\uparrow}_L]
{\lfoct{\Gamma}{\uslt{N}}{{\uparrow}A^+}{U}}
{\ifoct{\Gamma}{A^+}{N}{U}}
\qquad
\infer[{\supset}_L]
{\lfoct{\Gamma}{V; S}{A^+ \supset B^-}{U}}
{\rfoct{\Gamma}{V}{A^+}
 &
 \lfoct{\Gamma}{S}{B^-}{U}}
\]
\[
\mbox{\it (no rule $\top^-_L$)}
\qquad
\infer[\wedge^-_{L1}]
{\lfoct{\Gamma}{\pi_1; S}{A^- \wedge^- B^-}{U}}
{\lfoct{\Gamma}{S}{A^-}{U}}
\qquad
\infer[\wedge^-_{L2}]
{\lfoct{\Gamma}{\pi_2; S}{A^- \wedge^- B^-}{U}}
{\lfoct{\Gamma}{S}{B^-}{U}}
\]

\medskip
\fbox{$\stable{U}$}
\vspace{-12pt}
\[
\infer
{\stable{A^+}}
{}
\qquad
\infer
{\stable{\susp{A^-}}}
{}
\]
\caption{Proof terms for the focused sequent calculus.}
\label{fig:proofterms}
\end{figure}

While it is convenient and traditional to define a logic in terms of
rules, we follow Herbelin \shortcite{herbelin95lambda} in noting that it is 
sometimes 
easier to manipulate derivations using an appropriately-designed proof term
presentation of the logic. 
Our proof term language is primarily a
generalization of the {\it spine form} introduced by Cervesato and
Pfenning \shortcite{cervesato03linear}.\footnote{We also draw
  inspiration from the syntax of CLF \cite{watkins02concurrent},
  call-by-push-value \cite{levy04call}, and Modernized Algol
  \cite{harper12practical} for our syntax.}  Spine form is a proof
term assignment for the so-called uniform proofs, the focused fragment
of a logic that only includes the negative (or asynchronous)
propositions. In spine form, {\it terms} and {\it spines} correspond
to derivations of inversion sequents and left-focused sequents,
respectively; we also consider {\it values} corresponding to
derivations of right-focused sequents.
\[
\begin{array}{lrcl}
\mbox{\it Values}
 & V & ::= 
 & z \mid \dsrt{N} 
   \mid \mathsf{inl}\,V \mid \mathsf{inr}\,V
   \mid \langle\rangle^+ \mid \langle V_1, V_2 \rangle^+
\\
\mbox{\it Terms}
 & N, M & ::=
 & \rft{V} \mid \lft{x}{S}
   \mid \etalt{z}{N} \mid \dslt{x}{N}
   \mid \mathsf{abort} \mid [N_1, N_2]
   \mid \langle\rangle.N \mid \times N
\\ & & & 
   \mid \etart{N} \mid \usrt{N} \mid \lambda N
   \mid \langle\rangle^- \mid \langle N_1, N_2 \rangle^-
\\
\mbox{\it Spines}
 & S & ::=
 & \textsc{nil} \mid \uslt{N} \mid V; S \mid \pi_1; S \mid \pi_2; S
\end{array}
\]
The separation of our syntax into three categories corresponds to the
three-sequent view of our calculus. We will also refer to proof terms
generically as {\it expressions} $E$ when we want to invoke the
one-sequent view of our system.

Only two terms bind new variables. The term $\etalt{z}{N}$,
corresponding to the rule $\eta^+$, binds a positive variable $z$ (a
variable corresponding to a suspended proposition). The term
$\dslt{x}{N}$, corresponding to the rule ${\uparrow}_L$, binds a new
negative variable $x$ (a variable corresponding to a negative
proposition).  We will freely span the Curry-Howard correspondence (or
``propositions as types''), calling a value that corresponds to a
derivation of the right-focused sequent $\rfoc{\Gamma}{A^+}$ a value
focused on $A^+$ and calling a spine that corresponds to a derivation
of the left-focused sequent $\lfoc{\Gamma}{A^-}{B^+}$ or
$\lfoc{\Gamma}{A^-}{\susp{B^-}}$ a spine of type $B^+$ or $B^-$
(respectively) focused on $A^-$. Terms that correspond to stable
sequents $\ifoc{\Gamma}{\cdot}{A^+}$ and
$\ifoc{\Gamma}{\cdot}{\susp{A^-}}$ are (respectively) terms of type
$A^+$ or $A^-$.  Terms corresponding to derivations of more general
inversion sequents like $\ifoc{\Gamma}{\Omega}{A^+}$ and
$\ifoc{\Gamma}{\Omega}{A^-}$ are (respectively) terms of type $A^+$
introducing $\Omega$ and terms introducing $\Omega$ and $A^-$. We
reserve the word {\it type} for stable succedents to emphasize that
these are the important elements in the synthetic view of focusing.

It is possible to re-present the entire sequent calculus from
Figures~\ref{fig:foc}~and~\ref{fig:foc-atom} annotating sequents with
values, terms, and spines; the result is
Figure~\ref{fig:proofterms}. This ``Curry-style'' view, which sees
types as {\it extrinsic} to the proof terms, is helpful as a
reference, but it does not otherwise serve our purposes. Instead, we
will proceed with a ``Church-style'' view of types as {\it
  intrinsic}. This necessitates thinking of proof terms as carrying
some extra annotations; Pfenning writes these as superscripts
\cite{pfenning08church}, but which we will follow Girard in leaving
them implicit \cite{girard89proofs}.  In particular, positive
variables $z$ must be annotated with positive propositions,
negative variables $x$ must be annotated with negative propositions,
and $\mathsf{inl}$, $\mathsf{inr}$, $\pi_1$, and $\pi_2$ must be
annotated with the branch of the disjunction or conjunction that was
not taken.  This suffices to ensure that proof terms are in 1-to-1
correspondence with sequent calculus derivations modulo the structural
properties of exchange and weakening.\footnote{Our desire present cut
  admissibility and identity expansion using proof terms is one reason
  we use so many syntactic markers in our proof terms. These markers
  ($\times N$, $\etart{N}$, and $\usrt{N}$, etc.)~may be omitted in a
  Curry-style presentation.}

We will make a habit of presenting proof terms for admissible rules as
well. The admissible focal substitution principles labeled ${\it subst}^+$
and ${\it subst}^-$ above are respectively associated with the
functions $[V/z]E$ and $[E]S$ that act on proof terms:
\[
\infer-
{\efoct{\Gamma}{L}{[V/z]E}{U}}
{\rfoct{\Gamma}{V}{A^+}
 &
 \efoct{\Gamma, z{:}\susp{A^+}}{L}{E}{U}}
\qquad
\infer-
{\efoct{\Gamma}{L}{[E]S}{U}}
{\efoct{\Gamma}{L}{E}{\susp{A^-}}
 &
 \lfoct{\Gamma}{S}{A^-}{U}}
\]

\begin{figure}
\begin{tabbing}
\qquad \=
Proof term: \= $\lambda x_1.\lambda x_2.
  \langle \etart{\lft{x_1}{(\pi_1; \textsc{nil})}}, 
          \etart{\lft{x_2}{(\pi_1; \textsc{nil})}}\rangle^-$
\\
\>SML: \> \verb'fn x1 => fn x2 => (#1 x1, #1 x2)'
\\
\>Type: \> ${\downarrow}(p \wedge^- q) \supset
 {\downarrow}(r \wedge^- s) \supset 
 (p \wedge^- r)$
\\
\end{tabbing}

\begin{tabbing}
\qquad \=
Proof term: \= $\lambda {\times} z_1.z_2. 
\lambda {\times} z_3.z_4.
 \usrt{\rft{\langle z_1, z_3 \rangle^+}}$
\\
\>SML: \> \verb'fn (z1, z2) => fn (z3, z4) => (z1, z3)'
\\
\>Type: \> $(p \wedge^+ q) \supset 
 (r \wedge^+ s) \supset
 {\uparrow}(p \wedge^+ r)$
\\
\end{tabbing}

\begin{tabbing}
\qquad \=
Proof term: \= $\lambda f. \lambda g. 
    \lambda \etalt{z}{} \lambda[$\=
      $(\etalt{z_1}{\etart{\lft{f}{(z_1; z; \textsc{nil})}}}),$\\
\>\>\>$(\etalt{z_2}{\etart{\lft{g}{(\langle z_2, z\rangle^+; \textsc{nil})}}}]$
\\
\>SML: \> \verb'fn f => fn g => fn z => (fn Inl z1 => f z1 z'
\\   \>\> \verb'                          | Inr z2 => g (z2, z))'
\\
\>Type: \> ${\downarrow}(p^+ \supset s^+ \supset r^-) \supset
 {\downarrow}(q^+ \wedge^+ s^+ \supset r^-) \supset
 s^+ \supset (p^+ \vee q^+) \supset r^-$
\\
\end{tabbing}

\begin{tabbing}
\qquad \=
Proof term: \= $\lambda \etalt{z}{} \lambda f. \lambda g. 
 \{\lft{f}{(z; \uslt{[}}$\=$(
  \etalt{z_1}
        {\lft{g}{(z_1; \uslt{\etalt{z_3}{\rft{z_3}}})}}),$\\
\>\>\>$(\etalt{z_2}{\rft{z_2}})])\}$
\\
\>SML: \> \verb'fn z => fn f => fn g => (case (f z) of'
\\   \>\> \verb'                            Inl z1 => (case g z1 of z3 => z3)'
\\   \>\> \verb'                          | Inr z2 => z2)'
\\
\>Type: \> $p^+ \supset {\downarrow}(p^+ \supset {\uparrow}(q^+ \vee r^+)) \supset
 {\downarrow}(q^+ \supset {\uparrow}r^+) \supset {\uparrow}r^+$
\end{tabbing}
\caption{Some proof terms and their rough translation into Standard ML.}
\label{fig:sml-ex}
\end{figure}

\paragraph*{Patterns}
Our proof term calculus departs in one important way from most
presentations of focused proof terms. In other work, the trend is to
introduce the variables needed for an inversion phase all at once in a
syntactic entity called a {\it pattern}; one significant example is
Krishnaswami's presentation of ML-style pattern matching and pattern
compilation in the context of a focused sequent calculus
\cite{krishnaswami09focusing}. We do not use patterns because doing so
would not be faithful to the LF encoding of Figure~\ref{fig:foc} used
in the accompanying Twelf development; patterns cause the inductive
structure of proof terms and sequents to deviate, even if they remain
in 1-to-1 correspondence. 

While a full discussion of patterns is
beyond the scope of this article, we also want to suggest that our
choice is the natural one from the perspective of the sequent
calculus. Patterns are certainly relevant in the study of logic and
programming languages, but they seem more in line with natural
deduction presentations of logic or with higher-order focused
presentations, which can be seen as a synthesis of natural deduction
and sequent calculus presentations
\cite{zeilberger09logical,brock10focused}.

\paragraph*{Examples}
Using Standard ML's syntax as an imperfect proxy for a
natural-deduction system with pattern matching, we give, in
Figure~\ref{fig:sml-ex}, some proof terms and our suggestion as to the
corresponding natural deduction term. Note that if $M$ is a term
introducing $B^-$ then the proof term corresponding to
${\downarrow}A^- \supset B^-$ is $\lambda \dslt{x}{M}$, though in this
case the familiar construct is comprised of two smaller constructs,
the proof term corresponding to ${\supset}_R$ and the proof term
corresponding to ${\downarrow}_L$. The spine form $\uslt{M}$ comes
from Levy's CBPV, stands for {\it pattern match}, and corresponds to
\verb'case' in Standard ML.

\subsection{De-focalization}\label{sec:soundness}

We conclude this section by presenting the de-focalization property,
that $\ifoc{\Gamma}{\cdot}{U}$ implies
$\seq{(\Gamma)^\circledast}{(U)^\circledast}$, which we can prove independently
of any of the standard metatheoretic results for either system.
In order to generalize the induction hypothesis, we define
a new sequent form $\seq{\Gamma; \Psi}{P}$, where $\Psi$ is an ordered sequence 
that mimics the inversion context $\Omega$. The meaning of this 
sequent is defined
by two rules which force the ordered $\Psi$ context to introduce its contents
into the hypothetical context $\Gamma$ in a left-to-right order:
\[
\infer[\it cons]
{\seq{\Gamma; P, \Psi}{Q}}
{\seq{\Gamma, P; \Psi}{Q}}
\qquad
\infer[\it nil]
{\seq{\Gamma; \cdot}{Q}}
{\seq{\Gamma}{Q}}
\]

With this definition, we can state the appropriate generalization of
the induction hypothesis; our desired de-focalization property is a
corollary.

\begin{theorem}[De-focalization]\label{thm:soundness}
If $\efoc{\Gamma}{L}{U}$, then 
   $\seq{(\Gamma)^\circledast; (L)^\circledast}{(U)^\circledast}$
\end{theorem}
We can also state Theorem~\ref{thm:soundness} using the three-sequent
view of our logic. This statement of the theorem has three parts:
\begin{enumerate}
\item If $\rfoc{\Gamma}{A^+}$,
      then $\seq{(\Gamma)^\circledast; \cdot}{(A^+)^\bullet}$,
\item If $\ifoc{\Gamma}{\Omega}{U}$,
      then $\seq{(\Gamma)^\circledast; (\Omega)^\bullet}{(U)^\circledast}$, and
\item If $\lfoc{\Gamma}{A^-}{U}$, 
      then $\seq{(\Gamma)^\circledast; (A^-)^\bullet}{(U)^\circledast}$.
\end{enumerate}

\begin{proof}
By induction and case analysis on the given derivation;
$\mathcal D :: \seq{\Gamma}{P}$ denotes that $\mathcal D$ is a
derivation of $\seq{\Gamma}{P}$. Twenty-one of the twenty-four cases
are blindingly straightforward, such as this one:

\begin{description}
\item[Case]
$\mathcal D = \infer[\wedge^+_R]
{\rfoc{\Gamma}{A^+ \wedge B^+}}
{\deduce{\rfoc{\Gamma}{A^+}}{\mathcal D_1} 
 & 
 \deduce{\rfoc{\Gamma}{B^+}}{\mathcal D_2}}$

\begin{tabbing}
\qquad \= $\mathcal F_2'$ \= \kill
\>
$\mathcal E_1$ \> :: $\seq{(\Gamma)^\circledast; \cdot}{(A^+)^\bullet}$
 \` by the i.h.~(part 1) on $\mathcal D_1$
\\
\>
$\mathcal E_1'$ \> :: $\seq{(\Gamma)^\circledast}{(A^+)^\bullet}$
 \` by inversion on $\mathcal E_1$
\\
\>
$\mathcal E_2$ \> :: $\seq{(\Gamma)^\circledast; \cdot}{(B^+)^\bullet}$
 \` by the i.h.~(part 1) on $\mathcal D_2$
\\
\>
$\mathcal E_2'$ \> :: $\seq{(\Gamma)^\circledast}{(B^+)^\bullet}$
 \` by inversion on $\mathcal E_2$
\\
\>
$\mathcal E$ \> :: $\seq{(\Gamma)^\circledast}{(A^+)^\bullet \wedge (B^+)^\bullet}$
 \` by rule $\wedge_R$ on $\mathcal E_1$ and $\mathcal E_2$
\\
\>
$\mathcal E$ \> :: $\seq{(\Gamma)^\circledast}{(A^+ \wedge^+ B^+)^\bullet}$
 \` $(A^+  \wedge^+ B^+)^\bullet = (A^+)^\bullet \wedge (B^+)^\bullet$
\\
\>
$\mathcal E'$ \> :: $\seq{(\Gamma)^\circledast; \cdot}{(A^+ \wedge^+ B^+)^\bullet}$
 \` by rule ${\it nil}$ on $\mathcal E$.
\end{tabbing}
\end{description}

\noindent
For three cases corresponding to the rules $\bot_L$, $\vee_L$, and
$\wedge^+_L$, a secondary induction is needed to show the
admissibility, in the unfocused sequent calculus, of left rules that
have a context $\Psi$.

\begin{description}
\item[Case]
$\mathcal D = \infer[\wedge^+_L]
{\ifoc{\Gamma}{A^+ \wedge B^+, \Omega}{U}}
{\deduce{\ifoc{\Gamma}{A^+, B^+, \Omega}{U}}{\mathcal D_1}}$
\begin{tabbing}
\qquad \= $\mathcal F_2'$ \= \kill
\>
$\mathcal E_1$ \> :: $\seq{(\Gamma)^\circledast; (A^+, B^+, \Omega)^\bullet}{(U)^\circledast}$
 \` by i.h.~(part 2) on $\mathcal D_1$
\\
\>
$\mathcal E_1$ \> :: $\seq{(\Gamma)^\circledast; (A^+)^\bullet, (B^+)^\bullet, (\Omega)^\bullet}{(U)^\circledast}$
 \` $(A^+, B^+, \Omega)^\bullet = (A^+)^\bullet, (B^+)^\bullet, (\Omega)^\bullet$
\\
\>
$\mathcal E_1'$ \> :: $\seq{(\Gamma)^\circledast, (A^+)^\bullet; (B^+)^\bullet, (\Omega)^\bullet}{(U)^\circledast}$
 \` by inversion on $\mathcal E_1$ 
\\
\>
$\mathcal E_1''$ \> :: $\seq{(\Gamma)^\circledast, (A^+)^\bullet, (B^+)^\bullet; (\Omega)^\bullet}{(U)^\circledast}$
 \` by inversion on $\mathcal E_1'$ 
\\
\>
$\mathcal E$ \> :: $\seq{(\Gamma)^\circledast, (A^+)^\bullet \wedge (B^+)^\bullet; (\Omega)^\bullet}{(U)^\circledast}$
 \` by lemma on $\mathcal E''_1$
\\
\>
$\mathcal E$ \> :: $\seq{(\Gamma)^\circledast, (A^+ \wedge^+ B^+)^\bullet; (\Omega)^\bullet}{(U)^\circledast}$
 \` $(A^+  \wedge^+ B^+)^\bullet = (A^+)^\bullet \wedge (B^+)^\bullet$
\\
\>
$\mathcal E'$ \> :: $\seq{(\Gamma)^\circledast; (A^+ \wedge^+ B^+)^\bullet, (\Omega)^\bullet}{(U)^\circledast}$
 \` by rule $\it cons$ on $\mathcal E'$
\\
\>
$\mathcal E'$ \> :: $\seq{(\Gamma)^\circledast; (A^+ \wedge^+ B^+, \Omega)^\bullet}{(U)^\circledast}$
 \` $(A^+ \wedge^+ B^+, \Omega)^\bullet = (A^+ \wedge^+ B^+)^\bullet, (\Omega)^\bullet$
\end{tabbing}
\end{description}
The necessary lemma is that $\seq{\Gamma, P_1, P_2; \Psi}{Q}$,
implies $\seq{\Gamma, P_1 \wedge P_2; \Psi}{Q}$.
We proceed by induction on $\Psi$ and by case analysis on the
structure of the given derivation.
\begin{description}
\item[Subcase] $\mathcal D = 
\infer[\it cons]
{\seq{\Gamma, P_1, P_2; P, \Psi}{Q}}
{\deduce{\seq{\Gamma, P_1, P_2, P; \Psi}{Q}}{\mathcal D_1}}$
\begin{tabbing}
\qquad \= $\mathcal F_2'$ \= \kill
\>
$\mathcal D_1'$ \> :: $\seq{\Gamma, P, P_1, P_2; \Psi}{Q}$
 \` by exchange on $\mathcal D_1$\\
\>
$\mathcal E_1$ \> :: $\seq{\Gamma, P, P_1 \wedge P_2; \Psi}{Q}$
 \` by i.h.~on $\mathcal D_1'$\\
\>
$\mathcal E_1'$ \> :: $\seq{\Gamma, P_1 \wedge P_2, P; \Psi}{Q}$
 \` by exchange on $\mathcal E_1$\\
\>
$\mathcal E$ \> :: $\seq{\Gamma, P_1 \wedge P_2; P, \Psi}{Q}$
 \` by rule $\it cons$ on $\mathcal E_1'$
\end{tabbing}
\item[Subcase] $\mathcal D = 
\infer[\it nil]
{\seq{\Gamma, P_1, P_2; \cdot}{Q}}
{\deduce{\seq{\Gamma, P_1, P_2}{Q}}{\mathcal D_1}}$
\begin{tabbing}
\qquad \= $\mathcal F_2'$ \= \kill
\>
$\mathcal D_1'$ \> :: $\seq{\Gamma, P_1 \wedge P_2, P_1, P_2}{Q}$
 \` by weakening on $\mathcal D_1$\\
\>
$\mathcal E_1$ \> :: $\seq{\Gamma, P_1 \wedge P_2, P_1}{Q}$
 \` by rule $\wedge_{L2}$ on $\mathcal D_1'$\\
\>
$\mathcal E_1'$ \> :: $\seq{\Gamma, P_1 \wedge P_2}{Q}$
 \` by rule $\wedge_{L1}$ on $\mathcal E_1$\\
\>
$\mathcal E$ \> :: $\seq{\Gamma, P_1 \wedge P_2; \cdot}{Q}$
 \` by rule $\it nil$ on $\mathcal E_1'$
\end{tabbing}
\end{description}
The 22 other cases of the main theorem and the 2 other lemmas are
similar.  This theorem is named {\tt sound} in the accompanying Twelf
development.
\end{proof}

The lemma for $\wedge^+_L$ and the two similar lemmas for $\bot_L$ and
$\vee_L$ are as close as we will get to the tedious invertibility
lemmas encountered by other proofs of the focalization property.
Because of the way we have structured our system, each lemma only
requires induction and case analysis over the definition of
$\seq{\Gamma; \Psi}{P}$, which is defined by two rules, ${\it cons}$
and ${\it nil}$. Therefore, our proof remains linear in the number of
connectives and rules, rather than quadratic as in other approaches.

\section{Cut admissibility}
\label{sec:cut}

The statement of cut admissibility in an unpolarized logic is that
$\seq{\Gamma}{P}$ and $\seq{\Gamma, P}{Q}$ imply $\seq{\Gamma}{Q}$. In
polarized logic, we have positive and negative propositions, so the
cut admissibility theorem must, at minimum, have two parts: a negative
cut, that $\ifoc{\Gamma}{\cdot}{A^-}$ and $\ifoc{\Gamma,
  A^-}{\cdot}{U}$ imply $\ifoc{\Gamma}{\cdot}{U}$, and a positive cut,
that $\ifoc{\Gamma}{\cdot}{A^+}$ and $\ifoc{\Gamma}{A^+}{U}$ imply
$\ifoc{\Gamma}{\cdot}{U}$. The actual proof of cut admissibility will
require further generalization, but these statements are corollaries.
As in the statement of de-focalization, we state cut admissibility
using the one-sequent view of sequents in order to cut down on the
number of individual statements that we need to consider (4 parts
instead of the 7 we would need otherwise).

\begin{theorem}[Cut admissibility]
If $\Gamma$ and $U$ are suspension-normal, then\label{thm:cut}

\begin{enumerate}
\item If $\rfoc{\Gamma}{A^+}$ and $\ifoc{\Gamma}{A^+, \Omega}{U}$, then
      $\ifoc{\Gamma}{\Omega}{U}$,
\item If $\ifoc{\Gamma}{\cdot}{A^-}$, $\lfoc{\Gamma}{A^-}{U}$, 
      and $\stable{U}$, then
      $\ifoc{\Gamma}{\cdot}{U}$,
\item If $\ifoc{\Gamma}{\cdot}{A^-}$ and $\efoc{\Gamma, A^-}{L}{U}$, then 
      $\efoc{\Gamma}{L}{U}$, and
\item If $\efoc{\Gamma}{L}{A^+}$, $\ifoc{\Gamma}{A^+}{U}$,
      and $\stable{U}$,
      then $\efoc{\Gamma}{L}{U}$.
\end{enumerate}
\end{theorem}

\noindent
Beyond the additional cases needed to deal with shifts, the proof of
focused cut admissibility mirrors structural cut admissibility proofs
for unfocused sequent calculi.  In fact, the organization strategy
imposed by this four-part statement of cut admissibility makes
explicit the informal organization strategy of principal, 
left commutative, and right commutative cuts that Pfenning used to
present the many cases of structural cut admissibility proofs
\cite{pfenning00structural}.  

Before discussing the proof of Theorem~\ref{thm:cut}, we will show how
we write the four parts of cut admissibility at the level of proof
terms. By Curry-Howard, cut admissibility corresponds to a reduction
operation on proof terms that was named {\it hereditary substitution}
by Watkins et al.~\shortcite{watkins02concurrent}.

{\it Principal cuts} (parts 1 and 2) are cases where the {\it
  principal formula} (that is, $A^+$ or $A^-$) is the proposition
being decomposed in the last rule of both given derivations. In a
focused sequent calculus, this naturally happens when the principal
formula is in focus in one sequent and in inversion in the other. We
will refer to the operation of principal cuts on proof terms as {\it
  principal substitution}:
\[
\infer-[{\it cut}^+]
{\ifoct{\Gamma}{\Omega}{(V \bullet N)^{A^+}}{U}}
{\rfoct{\Gamma}{V}{A^+}
 &
 \ifoct{\Gamma}{A^+,\Omega}{N}{U}}
\]
\[
\infer-[{\it cut}^-]
{\ifoct{\Gamma}{\cdot}{(M \bullet S)^{A^-}}{U}}
{\ifoct{\Gamma}{\cdot}{M}{A^-}
 &
 \lfoct{\Gamma}{S}{A^-}{U}
 &
 \stable{U}
 }
\]

{\it Right commutative cuts} (part 3) deal with all cases where the
second given derivation decomposes a proposition other than the
principal formula. The action on proof terms is {\it rightist
  substitution}:
\[
\infer-[{\it rsubst}]
{\efoct{\Gamma}{L}{\llbracket M/x \rrbracket^{A^-} E}{U}}
{\ifoct{\Gamma}{\cdot}{M}{A^-}
 &
 \efoct{\Gamma, x{:}A^-}{L}{E}{U}}
\]

{\it Left commutative cuts} (part 4) deal with all cases where the
first given derivation ends in a left rule. The action on proof terms
is a {\it leftist substitution}: 
\[
\infer-[{\it lsubst}]
{\efoct{\Gamma}{L}{\llangle E \rrangle^{A^+} N}{U}}
{\efoct{\Gamma}{L}{E}{A^+}
 &
 \ifoct{\Gamma}{A^+}{N}{U}
 &
 \stable{U}
 }
\]

\begin{proof}
The proof of cut admissibility is by lexicographic induction. In each
invocation of the induction hypothesis, either
\begin{itemize}
\item the principal formula $A^+$ or $A^-$ 
gets smaller, or else it stays the same 
and
\item the ``part size'' (as in parts 1-4) decreases, or else both the
  principal formula and part size stay the same and either
\begin{itemize}
\item we are in part 3 and the second given derivation gets smaller,
  or
\item we are in part 4 and the first given derivation gets smaller.
\end{itemize}
\end{itemize}
This is actually a refinement of the standard structural induction
metric presented by Pfenning \shortcite{pfenning00structural}, which
is itself a structural-induction-flavored reinterpretation of the
metric used by Gentzen \shortcite{gentzen35untersuchungen} that forms
the basis of most cut elimination proofs. The
extra lexicographic ordering on ``part size'' is nonstandard, but is
needed here to justify the appeals to principal substitution
from rightist and leftist substitution.  When we look at the
computational content of cut admissibility, we can see that rightist
substitutions only break apart the second given derivation and that
leftist substitutions only break apart the first derivation, and that
these substitutions do not call one another directly. Unlike the usual
induction argument for cut admissibility, there is no commitment made
to the first derivation staying the same or getting smaller
while we are performing rightist substitution; the same is true for
the second derivation in leftist substitution. While it is beyond the
scope of this article, this alternate induction metric is helpful
when formalizing structural focalization in Agda.

Due to the conciseness (certainly) and clarity (optimistically) of
such a presentation, we present the cases of this proof using only
proof terms. This critically relies on the fact that we
understand all of our values, terms, and spines to be intrinsically
typed (and therefore in 1-to-1 correspondence with focused sequent
calculus derivations).

\subsubsection*{Principal substitution} 
This is where the action is; it's where both terms are decomposed
simultaneously in concert as the type gets smaller. Rightist and
leftist substitutions, in comparison, are just looking around for
places where principal substitution can happen.

\bigskip
\fbox{$(V \bullet N)^{A^+} = N'$} (part 1)
\begin{quote}
$(z \bullet \etalt{z'}{N})^{p^+} = [z/z']N$\\
$(\dsrt{M} \bullet x.N)^{{\downarrow}A^-}
   = \llbracket M/x \rrbracket^{A^-} N $\\
$(\mathsf{inl}\,V \bullet [N_1 , N_2 ])^{A^+ \vee B^+}
   = (V \bullet N_1)^{A^+}$\\
$(\mathsf{inr}\,V \bullet [N_1 , N_2 ])^{A^+ \vee B^+} 
   = (V \bullet N_2)^{B^+}$\\
$(\langle\rangle^+ \bullet \langle\rangle.N)^{\top^+} = N$\\
$(\langle V_1, V_2\rangle^+ \bullet {\times}N)^{A^+ \wedge^+ B^+} 
   = (V_2 \bullet (V_1 \bullet N)^{A^+})^{B^+}$

\medskip
In the case where $A^+ = p^+$, we invoke focal substitution $[z/z']N$
to do variable-for-variable substitution. This can also be seen as 
a use of contraction.
\end{quote}

\medskip
\fbox{$(M \bullet S)^{A^-} = N'$} (part 2)
\begin{quote}
$(\etart{M} \bullet \textsc{nil})^{p^-} = M$\\
$(\usrt{M} \bullet \uslt{N})^{{\uparrow}A^+} = \llangle M \rrangle^{A^+} N$\\
$(\lambda N \bullet V; S)^{A^+ \supset B^-} 
   = ((V \bullet N)^{A^+} \bullet S)^{B^-}$\\
$(\langle M_1, M_2 \rangle^- \bullet \pi_1; S)^{A^- \wedge^- B^-}
   = (M_1 \bullet S)^{A^-}$\\
$(\langle M_1, M_2 \rangle^- \bullet \pi_2; S)^{A^- \wedge^- B^-}
   = (M_2 \bullet S)^{B^-}$
\end{quote}

\subsubsection*{Rightist substitution} 
This is closest to the traditional form of substitution that we're
used to from natural deduction: we churn through the second term to
find all the places where $x$, the variable we're substituting $M$
for, occurs (if, indeed, any exist). When we find an occurrence of
this distinguished variable, which can only happen when the expression
that we're substituting into is a term that has decided to focus on
$x$, we call to negative principal substitutions (part 2).  In
traditional substitution we'd just plop $M$ down at the places where
$x$ occurred, but to do that in this setting would introduce a cut!

\bigskip
\fbox{$\llbracket M/x \rrbracket^{A^-} V = V'$} (part 3, $E = V$)
\begin{quote}
$\llbracket M/x \rrbracket^{A^-} z = z$\\
$\llbracket M/x \rrbracket^{A^-} \dsrt{N} 
   = \dsrt{(\llbracket M/x \rrbracket^{A^-} N)}$\\
$\llbracket M/x \rrbracket^{A^-} \mathsf{inl}\,V 
   = \mathsf{inl}\,(\llbracket M/x \rrbracket^{A^-} V)$\\
$\llbracket M/x \rrbracket^{A^-} \mathsf{inr}\,V 
   = \mathsf{inr}\,(\llbracket M/x \rrbracket^{A^-} V)$\\
$\llbracket M/x \rrbracket^{A^-} \langle\rangle^+ = \langle\rangle^+$\\
$\llbracket M/x \rrbracket^{A^-} \langle V_1, V_2 \rangle^+ 
   = \langle (\llbracket M/x \rrbracket^{A^-}V_1), 
             (\llbracket M/x \rrbracket^{A^-}V_2) \rangle^+$
\end{quote}

\medskip
\fbox{$\llbracket M/x \rrbracket^{A^-} N = N'$} (part 3, $E=N$)
\begin{quote}
$\llbracket M/x \rrbracket^{A^-} \rft{V}
   = \rft{(\llbracket M/x \rrbracket^{A^-} V)}$\\
$\llbracket M/x \rrbracket^{A^-} (\lft{x}{S})
   = (M \bullet \llbracket M/x \rrbracket^{A^-} S)^{A^-}$\\
$\llbracket M/x \rrbracket^{A^-} (\lft{x'}{S}) 
   = \lft{x'}{(\llbracket M/x \rrbracket^{A^-} S)}$
   \qquad (if $x \neq x'$)\\
$\llbracket M/x \rrbracket^{A^-} \etalt{z}{N} 
   = \etalt{z}{(\llbracket M/x \rrbracket^{A^-} N)}$\\
$\llbracket M/x \rrbracket^{A^-} \dslt{x'}{N} 
   = \dslt{x'}{(\llbracket M/x \rrbracket^{A^-} N)}$\\
$\llbracket M/x \rrbracket^{A^-} \mathsf{abort}
   = \mathsf{abort}$\\
$\llbracket M/x \rrbracket^{A^-} [N_1, N_2] 
   = [(\llbracket M/x \rrbracket^{A^-} N_1),
      (\llbracket M/x \rrbracket^{A^-} N_2)]$\\
$\llbracket M/x \rrbracket^{A^-} \langle\rangle.N
   = \langle\rangle.(\llbracket M/x \rrbracket^{A^-} N)$\\
$\llbracket M/x \rrbracket^{A^-} {\times}N 
   = {\times}(\llbracket M/x \rrbracket^{A^-} N)$\\
$\llbracket M/x \rrbracket^{A^-} \etart{N} 
   = \etart{\llbracket M/x \rrbracket N}$\\
$\llbracket M/x \rrbracket^{A^-} \usrt{N} 
   = \usrt{\llbracket M/x \rrbracket^{A^-} N}$\\
$\llbracket M/x \rrbracket^{A^-} \lambda N 
   = \lambda (\llbracket M/x \rrbracket^{A^-} N)$\\
$\llbracket M/x \rrbracket^{A^-} \langle\rangle^- = \langle\rangle^-$\\
$\llbracket M/x \rrbracket^{A^-} \langle N_1, N_2 \rangle^- 
   = \langle (\llbracket M/x \rrbracket^{A^-} N_1),
             (\llbracket M/x \rrbracket^{A^-} N_2) \rangle^-$

\medskip
In the cases for $\eta^+$ (proof term $\etalt{z}{N}$) and ${\downarrow}_L$
(proof term $\dslt{x'}{N}$), the bound variables
$z$ and $x'$ can always be 
$\alpha$-converted to be different from both $x$ and any variables free in $M$.
\end{quote}

\medskip
\fbox{$\llbracket M/x \rrbracket^{A^-} S = S'$} (part 3, $E = S$)

\begin{quote}
$\llbracket M/x \rrbracket^{A^-} \textsc{nil} = \textsc{nil}$\\ 
$\llbracket M/x \rrbracket^{A^-} \uslt{N} 
   = \uslt{(\llbracket M/x \rrbracket^{A^-} N)}$\\
$\llbracket M/x \rrbracket^{A^-} V; S 
   = (\llbracket M/x \rrbracket^{A^-} V); 
     (\llbracket M/x \rrbracket^{A^-} S)$\\
$\llbracket M/x \rrbracket^{A^-} \pi_1; S 
   = \pi_2; (\llbracket M/x \rrbracket^{A^-} S)$\\
$\llbracket M/x \rrbracket^{A^-} \pi_2; S 
   = \pi_1; (\llbracket M/x \rrbracket^{A^-} S)$
\end{quote}

\subsubsection*{Leftist substitution} 
This is so named because it, rather unusually, breaks apart the first
(and not the second) derivation. This is natural from the perspective
of cut elimination: the second term $N$ has an inversion it must do on
the left, so just like we searched in rightist substitution for any
(potential) use of the ${\it foc}_L$ rule on $x$ in the second term, we
search in leftist substitution for uses of ${\it foc}_R$ to derive
$A^+$ in the first term.

\bigskip
\fbox{$\llangle M \rrangle^{A^+} N = M'$} (part 4, $E=M$)
\begin{quote}
$\llangle \rft{V} \rrangle^{A^+} N 
   = (V \bullet N)^{A^+}$\\
$\llangle \lft{x}{S} \rrangle^{A^+} N  
   = \lft{x}{(\llangle S \rrangle^{A^+} N)}$\\
$\llangle \etalt{z}{M} \rrangle^{A^+} N  
   = \etalt{z}{(\llangle M \rrangle^{A^+} N)}$\\
$\llangle \dslt{x}{M} \rrangle^{A^+} N  
   = \dslt{x}{(\llangle M \rrangle^{A^+} N)}$\\
$\llangle \mathsf{abort} \rrangle^{A^+} N = \mathsf{abort}$\\
$\llangle [M_1, M_2] \rrangle^{A^+} N  
   = [(\llangle M_1 \rrangle^{A^+} N), 
      (\llangle M_2 \rrangle^{A^+} N)]$\\
$\llangle \langle\rangle.M \rrangle^{A^+} N  
   = \langle\rangle.(\llangle M \rrangle^{A^+} N)$\\
$\llangle {\times}M \rrangle^{A^+} N  
   = {\times}(\llangle M \rrangle^{A^+} N)$
\end{quote}

\medskip
\fbox{$\llangle S \rrangle^{A^+} N = S'$} (part 4, $E = S$)
\begin{quote}
$\llangle \uslt{M} \rrangle N
   = \uslt{(\llangle M \rrangle N)}$\\
$\llangle V; S \rrangle N
   = V; (\llangle S \rrangle N)$\\
$\llangle \pi_1; S \rrangle N
   = \pi_1; (\llangle S \rrangle N)$\\
$\llangle \pi_2; S \rrangle N
   = \pi_2; (\llangle S \rrangle N)$
\end{quote}

\noindent
This completes the proof. The four parts of this theorem are named
{\tt cut+}, {\tt cut-}, {\tt rsubst}, and {\tt lsubst} (respectively)
in the accompanying Twelf development.
\end{proof}

\section{Identity expansion}
\label{sec:expansion}

A significant novelty of our presentation relative to existing work is
our presentation of the identity expansion theorem; it is adapted from
the identity expansion theorem given for {\it weak focusing}
\cite{simmons11weak}, a less-restricted focusing calculus that does
not require invertible rules to be applied eagerly. The familiar
identity property for an unfocused sequent calculus states that, for
all propositions $A$, there is a derivation $\seq{\Gamma,
  A}{A}$. Identity in an unfocused sequent calculus can generally be
established by structural induction on the proposition $A$.

As with cut admissibility, there are two analogous identity properties
for the focused sequent calculus. First, for all positive propositions
$A^+$ there is a derivation of $\ifoc{\Gamma}{A^+}{A^+}$. Second, for
all negative propositions $A^-$ there is a derivation $\ifoc{\Gamma,
  A^-}{\cdot}{A^-}$. As an exercise, you should convince yourself that
this property cannot be established directly by structural induction
on $A^+$ or $A^-$. It doesn't work, in other words, to generalize the
${\it init}$ rule from the unfocused sequent calculus
(Figure~\ref{fig:unfoc}) to get an identity principle for the focused
sequent calculus. Instead, it is the suggestively named $\eta^+$ and
$\eta^-$ rules that generalize to admissible identity expansion
principles:
\[
\infer-[{\it expand}^+]
{\ifoc{\Gamma}{A^+, \Omega}{U}}
{\ifoc{\Gamma, \susp{A^+}}{\Omega}{U}}
\qquad
\infer-[{\it expand}^-]
{\ifoc{\Gamma}{\cdot}{A^-}}
{\ifoc{\Gamma}{\cdot}{\susp{A^-}}}
\]
When we introduced the $\eta^+$ and $\eta^-$ rules, we said they
reflected the idea that inversion should not suspend itself until
reaching an atomic proposition. The existence of these admissible
rules relaxes this requirement: we can optionally suspend inversion
before the pattern matching process is exhausted. The non-atomic
suspended propositions that appear when we suspend early appear to
have a connection to the {\it complex values} in call-by-push-value
\cite{levy04call}.

The premises of both of these rules are definitely not
suspension-normal. Unlike cut admissibility, identity expansion is not
at all restricted to suspension-normal sequents: non-atomic suspended
propositions and focal substitution play an important role. 

We associate positive identity expansion with the proof term $\eta^{A+}(z.N)$
and negative identity expansion with the proof term $\eta^{A^-}(N)$, 
allowing us to annotate the admissible rules above:
\[
\infer-[{\it expand}^+]
{\ifoct{\Gamma}{A^+, \Omega}{\eta^{A^+}(z.N)}{U}}
{\ifoct{\Gamma, z{:}\susp{A^+}}{\Omega}{N}{U}}
\qquad
\infer-[{\it expand}^-]
{\ifoct{\Gamma}{\cdot}{\eta^{A^-}(N)}{A^-}}
{\ifoct{\Gamma}{\cdot}{N}{\susp{A^-}}}
\]
Given identity expansion, the positive identity principle that
$\ifoc{\Gamma}{A^+}{A^+}$ holds for all $A^+$ is provable using
positive identity expansion.
\[
\infer-[{\it expand}^+]
{\ifoct{\Gamma}{A^+}{\eta^{A^+}(z.\rft{z})}{A^+}}
{\infer[{\it foc}_R]
 {\ifoct{\Gamma, z{:}\susp{A^+}}{\cdot}{\rft{z}}{A^+}}
 {\infer[{\it id}^+]
  {\rfoct{\Gamma, z{:}\susp{A^+}}{{z}}{A^+}}
  {}}}
\]
 The negative identity principle that $\ifoc{\Gamma, A^-}{\cdot}{A^-}$
 holds for all $A^-$ is similarly a corollary of negative identity
 expansion.
\[
\infer-[{\it expand}^-]
{\ifoct{\Gamma, x{:}A^-}{\cdot}{\eta^{A^-}(\lft{x}{\textsc{nil}})}{A^-}}
{\infer[{\it foc}_L]
 {\ifoct{\Gamma, x{:}A^-}{\cdot}{\lft{x}{\textsc{nil}}}{\susp{A^-}}}
 {\infer[{\it id}^-]
  {\lfoct{\Gamma, x{:}A^-}{\textsc{nil}}{A^-}{\susp{A^-}}}
  {}}}
\]

\begin{theorem}[Identity expansion]\label{thm:identity}~
\begin{enumerate}
\item For all $A^+$, if $\ifoc{\Gamma, \susp{A^+}}{\Omega}{U}$, 
then $\ifoc{\Gamma}{A^+, \Omega}{U}$. 
\item For all $A^-$, if $\ifoc{\Gamma}{\cdot}{\susp{A^-}}$,
then $\ifoc{\Gamma}{\cdot}{A^-}$. 
\end{enumerate}
\end{theorem}

\begin{proof}
The proof is by induction and case analysis on the structure of the
proposition $A^+$ or $A^-$.  

We will present one case of part 1 and
one case of part 2 line-by-line, and then present all of the cases
using the language of proof terms.

\begin{description}
\item[Case (part 1)] $A^+ = A^+ \wedge^+ B^+$
\begin{tabbing}
\qquad \= $\mathcal F_2'$ \= \kill
\>
$\mathcal D$ \> :: $\ifoc{\Gamma, \susp{A^+ \wedge^+ B^+}}{\Omega}{U}$
  \` given \\
\>
$\mathcal D'$ \> :: $\ifoc{\Gamma, \susp{A^+}, \susp{B^+}, \susp{A^+ \wedge^+ B^+}}{\Omega}{U}$
  \` by weakening on $\mathcal D$\\
\>
$\mathcal E_1$ \> :: $\rfoc{\Gamma, \susp{A^+}, \susp{B^+}}{A^+}$
  \` by rule ${\it id}^+$\\
\>
$\mathcal E_2$ \> :: $\rfoc{\Gamma, \susp{A^+}, \susp{B^+}}{B^+}$
  \` by rule ${\it id}^+$\\
\>
$\mathcal E$ \> :: $\rfoc{\Gamma, \susp{A^+}, \susp{B^+}}{A^+ \wedge^+ B^+}$
  \` by rule $\wedge^+_R$ on $\mathcal E_1$ and $\mathcal E_2$\\
\>
$\mathcal F$ \> :: $\ifoc{\Gamma, \susp{A^+}, \susp{B^+}}{\Omega}{U}$
  \` by focal substitution on $\mathcal E$ and $\mathcal D'$\\
\>
$\mathcal F_1$ \> :: $\ifoc{\Gamma, \susp{A^+}}{B^+, \Omega}{U}$
  \` by i.h.~(part 1) on $B^+$ and $\mathcal F$\\
\>
$\mathcal F_2$ \> :: $\ifoc{\Gamma}{A^+, B^+, \Omega}{U}$
  \` by i.h.~(part 1) on $A^+$ and $\mathcal F_1$\\
\>
$\ifoc{\Gamma}{A^+ \wedge^+ B^+, \Omega}{U}$
  \` by rule $\wedge^+_L$ on $F_2$
\end{tabbing}

\item[Case (part 2)] $A^- = A^+ \supset^- B^-$
\begin{tabbing}
\qquad \= $\mathcal F_2'$ \= \kill
\>
$\mathcal D$ \> :: $\ifoc{\Gamma}{\cdot}{\susp{A^+ \supset B^-}}$
  \` given \\
\>
$\mathcal D'$ \> :: $\ifoc{\Gamma, \susp{A^+}}{\cdot}{\susp{A^+ \supset B^-}}$
  \` by weakening on $\mathcal D$\\
\>
$\mathcal E_1$ \> :: $\rfoc{\Gamma, \susp{A^+}}{A^+}$
  \` by rule ${\it id}^+$\\
\>
$\mathcal E_2$ \> :: $\lfoc{\Gamma, \susp{A^+}}{B^-}{\susp{B^-}}$
  \` by rule ${\it id}^-$\\
\>
$\mathcal E$ \> :: $\lfoc{\Gamma,  \susp{A^+}}{A^+ \supset B^-}{\susp{B^-}}$
  \` by rule $\supset_L$ on $\mathcal E_1$ and $\mathcal E_2$\\
\>
$\mathcal F$ \> :: $\ifoc{\Gamma, \susp{A^+}}{\cdot}{\susp{B^-}}$
  \` by focal substitution on $\mathcal D'$ and $\mathcal E$\\
\>
$\mathcal F_1$ \> :: $\ifoc{\Gamma, \susp{A^+}}{\cdot}{{B^-}}$
  \` by i.h.~(part 2) on $B^-$ and $\mathcal F$\\
\>
$\mathcal F_2$ \> :: $\ifoc{\Gamma}{A^+}{{B^-}}$
  \` by i.h.~(part 1) on $A^-$ and $\mathcal F_1$\\
\>
$\ifoc{\Gamma}{\cdot}{A^+ \supset B^-}$
  \` by rule $\supset_R$ on $\mathcal F_2$
\end{tabbing}
\end{description}
This suffices to show the line-by-line structure of the identity
expansion theorem; other cases follow the same pattern. We will now
give all the cases on the level of proof terms:

\bigskip
\fbox{$\eta^{A^+}(v.N) = N'$} (part 1) 
\begin{quote}
$\eta^{p^+}(z.N) = \etalt{z}{N}$\\
$\eta^{{\downarrow}A^-}(z.N) 
  = \dslt{x}{([\dsrt{(\eta^{A^-}(\lft{x}{\textsc{nil}}))}/z]N)}$\\
$\eta^\bot(z.N) = \mathsf{abort}$\\
$\eta^{A^+ \vee B^+}(z.N) 
  = [\eta^{A^+}(z_1.[\mathsf{inl}\,z_1/z]N), 
     \eta^{B^+}(z_2.[\mathsf{inr}\,z_2/z]N)]$\\
$\eta^{\top^+}(z.N) = \langle\rangle.([\langle\rangle^+/z]N)$\\
$\eta^{A^+ \wedge^+ B^+}(z.N) 
  = {\times}(\eta^{A^+}(z_1.
            (\eta^{B^+}(z_2. [\langle z_1, z_2 \rangle^+/z]N))))$
\end{quote}

\medskip
\clearpage\fbox{$\eta^{A^-}(N) = N'$} (part 2) 
\begin{quote}
$\eta^{p^-}(N) = \etart{N}$\\
$\eta^{{\uparrow}A^+}(N) = \usrt{[N](\uslt{(\eta^{A^+}(z.\rft{z}))})}$\\
$\eta^{A^+ \supset B^-}(N)
  = \lambda(\eta^{A^+}(z. (\eta^{B^-}([N](z; \textsc{nil})))))$\\
$\eta^\top(N) = \langle\rangle^-$\\
$\eta^{A^- \wedge^- B^-}(N) 
  = \langle \eta^{A^-}([N](\pi_1;\textsc{nil})), 
            \eta^{B^-}([N](\pi_2;\textsc{nil}))\rangle^-$
\end{quote}

\noindent
This completes the proof; the two parts of this theorem are named
{\tt expand+} and {\tt expand-} (respectively) 
in the accompanying Twelf development.
\end{proof}

\section{Focalization}
\label{sec:focalization}

Theorem~\ref{thm:completeness} in this section establishes the
focalization property: it is possible to turn the 
unfocused derivation of an unpolarized sequent into a focused derivation
for any polarized sequent that erases to the unpolarized one. This
proof naturally factors into two parts. The first part is a series of
{\it unfocused admissibility} lemmas, a family of admissible rules
which serve to show that focused sequent calculus derivations can
mimic unfocused derivations.  The second part is a straightforward
inductive proof that, if $U$ is stable,
$\seq{(\Gamma)^\circledast}{(U)^\circledast}$ implies
$\ifoc{\Gamma}{\cdot}{U}$. Recall that $(-)^\circledast$, from
Figure~\ref{fig:erasure-ctx}, is the erasure of polarization for
contexts and succedents.

\subsection{Unfocused admissibility}\label{sec:unfocusedadmissibility}

We think of unfocused admissibility as building an
abstraction layer on top of focused, polarized logic.  The proof of
focalization then interacts with focused derivations entirely through
the abstraction layer of unfocused admissibility. 

It is possible to motivate unfocused admissibility independently of
focalization.  Consider the unfocused right rules for
conjunction compared to the focused right rules for (positive)
conjunction.
\[
\infer[\wedge_R]
{\seq{\Gamma}{A \wedge B}}
{\seq{\Gamma}{A} & \seq{\Gamma}{B}}
\qquad
\infer[\wedge^+_R]
{\rfoc{\Gamma}{A^+ \wedge^+ B^+}}
{\rfoc{\Gamma}{A^+} & \rfoc{\Gamma}{B^+}}
\]
The rules look similar, but their usage is quite different. To prove
$A \wedge B$ we must prove $A$ (possibly doing some work on the left first)
and, in the other branch, we must prove $B$ (possibly doing some work on the
left first). To prove $A^+ \wedge^+ B^+$, we must 
decompose $A^+$ in one branch and $B^+$ in the other; there is no possibility
of doing work on the left first. The admissible rule in polarized
logic that actually matches the structure of the unfocused rule $\wedge_R$
looks like this:\footnote{The admissible rules associated with 
the lemmas in this section will all be
annotated with a $u$ for {\it unfocused} (e.g. $\wedge^+_{uR}$).}
\[
\infer-[\wedge^+_{uR}]
{\ifoc{\Gamma}{\cdot}{A^+ \wedge^+ B^+}}
{\ifoc{\Gamma}{\cdot}{A^+} & 
 \ifoc{\Gamma}{\cdot}{B^+}
}
\]
The stable premises $A^+$ and $B^+$ ensure that, in both
subderivations, it will be possible to do work on the left before
decomposing $A^+$ or $B^+$. 

The unfocused admissibility lemmas could be established the slow,
painful, and boring way, by one or more inductions over focused
derivations per lemma.
This more traditional approach is
both technically and philosophically unsatisfying, however. The approach is
technically unsatisfying because these theorems are long and annoying,
and it is philosophically unsatisfying because cut admissibility and
identity expansion are already supposed to capture global properties
of the logic. We will instead establish unfocused admissibility
directly from cut admissibility and identity expansion without the
need for any additional induction; each unfocused
admissibility proof is short, though dense. In Figure~\ref{fig:adm} we
present the proof of $\wedge^+_{uR}$ as a derivation built using
admissible rules.

\begin{figure}
\[
\infer-[{\it rsubst}]
{\ifoc{\Gamma}{\cdot}{A^+ \wedge^+ B^+}}
{\infer[{\uparrow}_R]
 {\ifoc{\Gamma}{\cdot}{{\uparrow}A^+}}
 {\deduce
  {\ifoc{\Gamma}{\cdot}{A^+}}
  {\mathcal D_1}}
 &
 \infer-[{\it lsubst}]
 {\ifoc{\Gamma, {\uparrow}A^+}{\cdot}{A^+ \wedge^+ B^+}}
 {\infer-[{\it weaken}]
  {\ifoc{\Gamma, {\uparrow}A^+}{\cdot}{B^+}}
  {\deduce
   {\ifoc{\Gamma}{\cdot}{B^+}}
   {\mathcal D_2}}
  & 
  \!\!\!\!\!
  \infer-[{\it expand}^+]
  {\ifoc{\Gamma, {\uparrow}A^+}{B^+}{A^+ \wedge^+ B^+}} 
  {\infer[{\it foc}_L]
   {\ifoc{\Gamma, {\uparrow}A^+, \susp{B^+}}{\cdot}
         {A^+ \wedge^+ B^+}}
   {\infer[{\uparrow}_L]
    {\ifoc{\Gamma, {\uparrow}A^+, \susp{B^+}}{[{\uparrow}A^+]}
          {A^+ \wedge^+ B^+}}
    {\infer-[{\it expand}^+]
     {\ifoc{\Gamma, {\uparrow}A^+, \susp{B^+}}{A^+}
           {A^+ \wedge^+ B^+}}
     {\infer[{\it foc}_R]
      {\ifoc{\Gamma'}
            {\cdot}
            {A^+ \wedge^+ B^+}}
      {\infer[\wedge^+_R]
       {\rfoc{\Gamma'}
             {A^+ \wedge^+ B^+}}
       {\infer[{\it id}^+]
        {\rfoc{\Gamma'}
              {A^+}}
        {}
        &
        \infer[{\it id}^+]
        {\rfoc{\Gamma'}
              {B^+}}
        {}}}}}}}}}
\]
\caption{Unfocused admissibility rule $\wedge^+_{uR}$ as a derivation,
  where $\Gamma' = \Gamma, {\uparrow}A^+, \susp{B^+}, \susp{A^+}$.}
\label{fig:adm}
\end{figure}


The unfocused admissibility lemmas will be presented in terms of the
admissible rules they justify, but their proofs will be presented
entirely at the level of proof terms. In most cases, we will omit the
propositions that annotate instances of cut admissibility.  We must be
careful about the interaction of cut admissibility and identity
expansion. The premises of ${\it expand}^+$ and ${\it expand}^-$ are
not suspension-normal because they contain non-atomic suspended
propositions, and cut admissibility is only defined on
suspension-normal sequents. All the unfocused admissibility lemmas in
this section require that the given sequents are stable and
suspension-normal, but we omit this repetitive precondition when we
write the admissible rules.\footnote{For lemmas that do not use cut
  admissibility, stability and suspension-normality are usually
  unnecessary preconditions. The mechanized proof states these less
  restrictive preconditions where they apply.}

\newcommand{\deshift}{\mbox{${\uparrow}{\downarrow}\hspace{-9.2pt}\diagup$}}

In certain cases we do more work than necessary, such as in the left
rule for ${\uparrow}{\top^+}$, which could alternatively be phrased as
a use of weakening. This is done to match the structure of unfocused
admissibility in substructural logics \cite{simmons12substructural}.

\subsubsection{Initial rules} 

A positive atomic proposition can appear in the hypothetical context
either as a shifted positive proposition $x{:}{\uparrow}p^+$ or as a
suspended positive proposition $z{:}\susp{p^+}$, and likewise for
negative atomic propositions on the right. As a result, we need four
initial rules to correspond to the single unfocused rule ${\it init}$.
(We could cut these four rules down to two if we restricted erasure
and focalization to suspension-free sequents instead of
suspension-normal sequents.)  Each of these unfocused admissibility
lemmas are actually directly derivable.

\[
\infer-[{\it initsusp}_u^-]
{\ifoct{\Gamma, x{:}p^-}{\cdot}{{\it initsusp}_u^-(x)}{\susp{p^-}}{}}
{}
\]

$\quad{\it initsusp}_u^-(x) = \lft{x}{\textsc{nil}}$

\[
\infer-[{\it init}_u^-]
{\ifoct{\Gamma, x{:}p^-}{\cdot}{{\it init}_u^-(x)}{{\downarrow}p^-}{}}
{}
\]

$\quad{\it init}_u^-(x) = \rft{(\dsrt{\etart{\lft{x}{\textsc{nil}}}})}$

\[
\infer-[{\it initsusp}_u^+]
{\ifoct{\Gamma, z{:}\susp{p^+}}{\cdot}{{\it initsusp}_u^+(z)}{p^+}}
{}
\]

$\quad{\it initsusp}_u^+(z) = \rft{z}$

\[
\infer-[{\it init}_u^+]
{\ifoct{\Gamma, x{:}{\uparrow}p^+}{\cdot}{{\it init}_u^+(x)}{p^+}}
{}
\]

$\quad{\it init}_u^+(x) = \lft{x}{\uslt{(\etalt{z}{\rft{z}})}}$

\subsubsection{Disjunction}

\[
\infer-[\bot_{uL}]
{\ifoct{\Gamma, x{:}{\uparrow}\bot}{\cdot}{{\bot}_{uL}(x)}{U}}
{}
\]

$\quad{\bot}_{uL}(x) = \lft{x}{(\uslt{\mathsf{abort}})}$

\[
\infer-[\vee_{uR1}]
{\ifoct{\Gamma}{\cdot}{\vee_{uR1}(N_1)}{A^+ \vee B^+}}
{\ifoct{\Gamma}{\cdot}{N_1}{A^+}}
\]

$\quad\vee_{uR1}(N_1) = 
   \llangle N_1 \rrangle^{A^+}
     (\eta^{A^+}(z.\rft{(\mathsf{inl}\,z)}))$

\[
\infer-[\vee_{uR2}]
{\ifoct{\Gamma}{\cdot}{\vee_{uR2}(N_2)}{A^+ \vee B^+}}
{\ifoct{\Gamma}{\cdot}{N_2}{B^+}}
\]

$\quad\vee_{uR2}(N_2) = 
   \llangle N_2 \rrangle^{B^+}
     (\eta^{B^+}(z.\rft{(\mathsf{inr}\,z)}))$

\[
\infer-[\vee_{uL}]
{\ifoct{\Gamma, x{:}{\uparrow}(A^+ \vee B^+)}
       {\cdot}{\vee_{uL}(x, x_1.N_1, x_2.N_2)}{U}}
{\ifoct{\Gamma, x_1{:}{\uparrow}A^+}{\cdot}{N_1}{U}
 &
 \ifoct{\Gamma, x_2{:}{\uparrow}B^+}{\cdot}{N_2}{U}}
\]

$\quad \vee_{uL}(x, x_1.N_1, x_2.N_2) = 
   \lft{x}{\uslt
     {(\llangle N_{\it Id} \rrangle 
        [\dslt{x_1}{N_1} ,
         \dslt{x_2}{N_2}])}}
  $

\smallskip
\quad 
where $N_{\it Id} =
  [ \eta^{A^+}(z_1. \rft{(\mathsf{inl}\,(\dsrt{\usrt{\rft{z_1}}}))}),
    \eta^{B^+}(z_2. \rft{(\mathsf{inr}\,(\dsrt{\usrt{\rft{z_1}}}))})
  ]$

\quad
is a closed term of type 
${\downarrow}({\uparrow}A^+) \vee {\downarrow}({\uparrow}B^+)$ 
introducing $A^+ \vee B^+$

\subsubsection{Positive conjunction}

\[
\infer-[\top^+_{uR}]
{\ifoct{\Gamma}{\cdot}{\top^+_{uR}}{\top^+}}
{}
\]

$\quad\top^+_{uR} = \rft{\langle\rangle^+}$

\[
\infer-[\top^+_{uL}]
{\ifoct{\Gamma, x{:}{\uparrow}\top^+}{\cdot}{\top^+_{uL}(x,N)}{U}}
{
 \ifoct{\Gamma}{\cdot}{N}{U}}
\]

$\quad\top^+_{uL}(x,N) = \lft{x}{\uslt{(\langle\rangle. N)}}$
   %

\[
\infer-[\wedge^+_{uR}]
{\ifoct{\Gamma}{\cdot}{\wedge^+_{uR}(N_1,N_2)}{A^+ \wedge^+ B^+}}
{\ifoct{\Gamma}{\cdot}{N_1}{A^+}
 &
 \ifoct{\Gamma}{\cdot}{N_2}{B^+}}
\]

$\quad\wedge^+_{uR}(N_1,N_2) =
 \llbracket \usrt{N_1} / x_1 \rrbracket^{{\uparrow}A^+}
  ( \llangle N_2 \rrangle^{B^+}
    N_{\it Id}(x_1))
$

\smallskip
\quad 
where $N_{\it Id}(x_1) = 
  \eta^{B^+}(z_2. \lft{x_1}
    {\uslt{(\eta^{A^+}(z_1. \rft{\langle z_1, z_2 \rangle^+}))}})$

\quad
    is a term of type $A^+ \wedge^+ B^+$ introducing $B^+$ with 
    $x$ of type ${\uparrow}A^+$ free.

\quad(This was the case given above as a derivation with admissible rules.)

\[
\infer-[\wedge^+_{uL}]
{\ifoct{\Gamma,x{:}{\uparrow}(A^+ \wedge^+ B^+)}{\cdot}
       {\wedge^+_{uL}(x,x_1.x_2.N_1)}{U}}
{
 \ifoct{\Gamma,x_1{:}{\uparrow}A^+,x_2{:}{\uparrow}B^+}{\cdot}
       {N_1}{U}}
\]

$\quad\wedge^+_{uL}(x,x_1.x_2.N_1) = 
   \lft{x}{\uslt{(\llangle N_{\it Id} \rrangle ({\times}\dslt{x_1}{\dslt{x_2}{N_1}}))}}$

\smallskip
\quad 
where $N_{\it Id} = 
  {\times} 
     (\eta^{A^+}
       (z_1. \eta^{B^+}
         (z_2.
           \rft{
             \langle
               \dsrt{\usrt{\rft{z_1}}},
               \dsrt{\usrt{\rft{z_2}}}
             \rangle^+}
         )
       )
     )$

\quad
is a closed term of type 
${\downarrow}{\uparrow}A^+ \wedge^+ {\downarrow}{\uparrow}B^+$
introducing $A^+ \wedge^+ B^+$.

\subsubsection{Implication}

\[
\infer-[\supset_{uR}]
{\ifoct{\Gamma}{\cdot}{\supset_{uR}(x_1.N_1)}{{\downarrow}(A^+ \supset B^-)}}
{\ifoct{\Gamma, x_1{:}{\uparrow}A^+}{\cdot}{N_1}{{\downarrow}B^-}}
\]

$\quad\supset_{uR}(x_1.N_1) =
  \rft{(\dsrt{(\llbracket \lambda \dslt{x_1}{\usrt{N_1}}/x \rrbracket N_{\it Id}(x))})}
  $

\smallskip
\quad
where $N_{\it Id}(x) = 
  \lambda
    (\eta^{A^+}
      (z.\eta^{B^-}
        \lft{x}{(\dsrt{\usrt{\rft{z}}}); 
                (\uslt{\dslt{x'}{\lft{x'}{\textsc{nil}}}})}
      )
    )$

\quad
is a term introducing $A^+ \supset B^-$ 
with $x$ of type 
${\downarrow}{\uparrow}A^+ \supset {\uparrow}{\downarrow}B^-$ free.

\[
\infer-[\supset_{uL}]
{\ifoct{\Gamma, x{:}A^+ \supset B^-}{\cdot}{\supset_{uL}(N_1,x_2.N_2)}{U}}
{
 \ifoct{\Gamma}{\cdot}{N_1}{A^+}
 &
 \ifoct{\Gamma, x_2{:}B^-}{\cdot}{N_2}{U}}
\]

$\quad\supset_{uL}(x,N_1,x_2.N_2) =
   \llangle 
     \llangle
     N_1
     \rrangle^{A^+} N_{\it Id}(x)     
   \rrangle^{{\downarrow} B^-}
   \dslt{x_2}{N_2}
$

\smallskip
\quad
where $N_{\it Id}(x) = 
     \eta^{A^+}
      (z. \rft{(\dsrt{(\eta^{B^-}
       (\lft{x}{(z; \textsc{nil})}
       ))})}
      )
     $

\quad is a term of type ${\downarrow}B^-$ introducing $A^+$ with
$x$ of type $A^+ \supset B^-$ free.


\paragraph{Negative conjunction}

\[
\infer-[\top^-_{uR}]
{\ifoct{\Gamma}{\cdot}{\top^-_{uR}}{{\downarrow}\top^-}}
{}
\]

$\quad\top^-_{uR}= \rft{(\dsrt{(\langle\rangle^-)})}$

\[
\infer-[\wedge^-_{uR}]
{\ifoct{\Gamma}{\cdot}{\wedge^-_{uR}(N_1,N_2)}{{\downarrow}(A^- \wedge^- B^-)}}
{\ifoct{\Gamma}{\cdot}{N_1}{{\downarrow}A^-}
 &
 \ifoct{\Gamma}{\cdot}{N_2}{{\downarrow}B^-}}
\]

$\quad\wedge^-_{uR}(N_1,N_2) = 
  \rft{(\dsrt{(
    \llbracket \langle \usrt{N_1}, \usrt{N_2} \rangle^-/x \rrbracket 
    N_{\it Id}(x)
  )})}
$

\smallskip
\quad
where $N_{\it Id}(x) = 
  \langle
    \eta^{A^-}(\lft{x}{\pi_1; \uslt{\dslt{y}{(\lft{y}{\textsc{nil}})}}}),
    \eta^{B^-}(\lft{x}{\pi_2; \uslt{\dslt{y}{(\lft{y}{\textsc{nil}})}}})
  \rangle^-$

\quad
is a term introducing $A^- \wedge^- B^-$ 
with $x$ of type 
${\uparrow}{\downarrow}A^- \wedge^- {\uparrow}{\downarrow}B^-$ free.

\[
\infer-[\wedge^-_{uL1}]
{\ifoct{\Gamma, x{:}A^- \wedge^- B^-}{\cdot}{\wedge^-_{uL1}(x,x_1.N_1)}{U}}
{\ifoct{\Gamma, x_1{:}A^-}{\cdot}{N_1}{U}}
\]

$\quad\wedge^-_{uL1}(x,x_1.N_1) = 
 \llbracket \eta^{A^-}(\lft{x}{\pi_1; \textsc{nil}}) / x_1 \rrbracket^{A^-} N_1$

\[
\infer-[\wedge^-_{uL2}]
{\ifoct{\Gamma, x{:}A^- \wedge^- B^-}{\cdot}{\wedge^-_{uL2}(x,x_2.N_2)}{U}}
{\ifoct{\Gamma, x_2{:}B^-}{\cdot}{N_2}{U}}
\]

$\quad\wedge^-_{uL2}(x,x_2.N_2) = 
 \llbracket \eta^{B^-}(\lft{x}{\pi_2; \textsc{nil}}) / x_2 \rrbracket^{B^-} N_2$

\subsubsection{Shift removal} 

In order for the unfocused admissibility lemmas to form a complete
abstraction boundary between the focused sequent calculus and the
focalization theorem, we must account for the fact that many polarized
propositions erase to the same proposition. For example, if
$(A^+)^\bullet = P_1$ and $(B^-)^\bullet = P_2$, then
\[P_1 \supset P_2 = ({\downarrow}(A^+ \supset B^-))^\bullet = ({\downarrow}{\uparrow}{\downarrow}(A^+ \supset B^-))^\bullet = ({\downarrow}{\uparrow}{\downarrow}{\uparrow}{\downarrow}(A^+ \supset B^-))^\bullet = \ldots\]
and so on. To deal with deeply-shifted propositions in the
completeness theorem, we will invoke a shift removal lemma. It is
different from the other unfocused admissibility lemmas in that it
mentions erasure and we prove it by induction over the structure of
propositions.

\begin{lemma}[Shift removal (positive)]

If $(A^+)^\bullet = P$, there exists a $B^+$, 
  not of the form ${\downarrow}{\uparrow}C^+$, such that
$(B^+)^\bullet = P$ and, for any $\Gamma$, $\ifoc{\Gamma}{\cdot}{B^+}$
implies $\ifoc{\Gamma}{\cdot}{A^+}$.
\end{lemma}

\begin{lemma}[Shift removal (negative)] 

If $(A^-)^\bullet = P$, there exists a $B^-$, 
  not of the form ${\uparrow}{\downarrow}C^-$, such that
$(B^-)^\bullet = P$ and, for any $\Gamma$ and $\stable{U}$,
$\ifoc{\Gamma, B^-}{\cdot}{U}$ implies $\ifoc{\Gamma, A^-}{\cdot}{U}$
\end{lemma}

\begin{proof}
Both lemmas are by induction on the structure of the proposition $A^+$
or $A^-$. If the outermost structure of the proposition is two
adjacent shifts, we invoke the induction hypothesis and apply either
${\downarrow}{\uparrow}_{uR}$ or
${\uparrow}{\downarrow}_{uL}$:

\[
\infer-[\mathit{{\downarrow}{\uparrow}_{uR}}]
{\ifoct{\Gamma}{\cdot}
       {\mathit{{\downarrow}{\uparrow}_{uR}}(N_1)}
       {{\downarrow}{\uparrow}A^+}}
{\ifoct{\Gamma}{\cdot}{N_1}{A^+}}
\]

$\quad\mathit{{\downarrow}{\uparrow}_{uR}}(N_1) =
 \rft{(\dsrt{\usrt{N_1}})}$

\[
\infer-[\mathit{{\uparrow}{\downarrow}_{uL}}]
{\ifoct{\Gamma, x{:}{\uparrow}{\downarrow}A^-}{\cdot}
       {\mathit{{\uparrow}{\downarrow}_{uL}}(x,x_1.N_1)}{U}}
{\ifoct{\Gamma, x_1{:}A^-}{\cdot}{N_1}{U}}
\]

$\quad\mathit{{\uparrow}{\downarrow}_{uL}}(x,x_1.N_1) =
  \lft{x}{\uslt{\dslt{x_1}{N_1}}}$

\bigskip
\noindent
In all cases where the outermost structure of the proposition is {\em
  not} made up of two adjacent shifts, we succeed immediately using
the given derivation.  These lemmas are called {\tt rshifty} and {\tt
  lshifty} in the accompanying Twelf development.
\end{proof}

\subsection{Proof of focalization}
\label{sec:focalizationproof}

Since we have not defined proof terms corresponding to 
unfocused sequent calculus derivations, in
the proof of the focalization we will return to the more
traditional style of proof presentation.

\begin{theorem}[Focalization]\label{thm:completeness}
If $U$ is stable and $\Gamma$ and $U$ are suspension-normal, then
$\seq{(\Gamma)^\circledast}{(U)^\circledast}$ implies $\ifoc{\Gamma}{\cdot}{U}$.
\end{theorem}

\noindent
The second condition, that $\Gamma$ and $U$ are suspension-normal, is
not something we strictly need to state, as $(\Gamma)^\circledast$ and
$(U)^\circledast$ are only defined on suspension-normal contexts and
succedents.

\begin{proof}
By induction on the structure of the given derivation $\mathcal
D$. Aside from the rule ${\it init}$, each rule in
Figure~\ref{fig:unfoc} decomposes one proposition $P$ on the left or
the right. By the definition of erasure in
Figure~\ref{fig:erasure-ctx}, if $P$ is being decomposed on the right
then $P = (A^+)^\bullet$ for some $A^+$. We proceed by 
case analysis on the structure of the polarized proposition $A^+$.
By
the shift removal lemma, it suffices to consider the case where this
formula is not double-shifted. We will show a few representative
cases.

\begin{description}
\item[Case] $A^+ = \top^+$, \quad
$\mathcal D = \infer[\top_R]
{\seq{(\Gamma)^\circledast}{\top}}{}$
\begin{tabbing}
\qquad \= $\mathcal F_2'$ \= \kill
\>
$\ifoc{\Gamma}{\cdot}{\top^+}$ 
  \` by unfocused admissibility lemma $\top^+_{uR}$
\end{tabbing}

\item[Case] $A^+ = {\downarrow}\top^-$, \quad
$\mathcal D = \infer[\top_R]
{\seq{(\Gamma)^\circledast}{\top}}{}$
\begin{tabbing}
\qquad \= $\mathcal F_2'$ \= \kill
\>
$\ifoc{\Gamma}{\cdot}{{\downarrow}\top^-}$ 
  \` by unfocused admissibility lemma $\top^-_{uR}$
\end{tabbing}

\item[Case] $A^+ = B_1^+ \wedge^+ B_2^+$, \quad
$\mathcal D = \infer[\wedge_R]
{\seq{(\Gamma)^\circledast}{(B_1^+)^\bullet \wedge (B_2^+)^\bullet}}
{\deduce{\seq{(\Gamma)^\circledast}{(B_1^+)^\bullet}}{\mathcal D_1}
 &
 \deduce{\seq{(\Gamma)^\circledast}{(B_2^+)^\bullet}}{\mathcal D_2}}$
\begin{tabbing}
\qquad \= $\mathcal F_2'$ \= \kill
\>
$\mathcal E_1$ \> :: $\ifoc{\Gamma}{\cdot}{B_1^+}$
  \` by i.h.~on $\mathcal D_1$\\
\>
$\mathcal E_2$ \> :: $\ifoc{\Gamma}{\cdot}{B_2^+}$
  \` by i.h.~on $\mathcal D_2$\\
\>
$\ifoc{\Gamma}{\cdot}{B_1^+ \wedge^+ B_2^+}$ 
  \` by unfocused admissibility lemma $\wedge^+_{uR}$ 
     on $\mathcal E_1$ and $\mathcal E_2$
\end{tabbing}

\item[Case] $A^+ = {\downarrow}(B_1^- \wedge^- B_2^-)$, \quad
$\mathcal D = \infer[\wedge_R]
{\seq{(\Gamma)^\circledast}{(B_1^-)^\bullet \wedge (B_2^-)^\bullet}}
{\deduce{\seq{(\Gamma)^\circledast}{(B_1^-)^\bullet}}{\mathcal D_1}
 &
 \deduce{\seq{(\Gamma)^\circledast}{(B_2^-)^\bullet}}{\mathcal D_2}}$
\begin{tabbing}
\qquad \= $\mathcal F_2'$ \= \kill
\>
$\mathcal E_1$ \> :: $\ifoc{\Gamma}{\cdot}{{\downarrow}B_1^-}$
  \` by i.h.~on $\mathcal D_1$\\
\>
$\mathcal E_2$ \> :: $\ifoc{\Gamma}{\cdot}{{\downarrow}B_2^-}$
  \` by i.h.~on $\mathcal D_2$\\
\>
$\ifoc{\Gamma}{\cdot}{{\downarrow}(B_1^- \wedge^- B_2^-)}$ 
  \` by unfocused admissibility lemma $\wedge^-_{uR}$ 
     on $\mathcal E_1$ and $\mathcal E_2$
\end{tabbing}
\end{description}
There are three other cases
corresponding to $\vee_{R1}$, $\vee_{R2}$, and $\supset_R$.
All proceed in a similar fashion. 

Similarly, if $P$ is being decomposed
on the left, then $P = (A^-)^\bullet$ for some $A^-$, and we proceed
using the shift removal lemma and case analysis on the structure
of $A^-$.

\begin{description}
\item[Case] $A^- = {\uparrow}(B_1^+ \wedge^+ B_2^+)$, \quad
$\mathcal D = \infer[\wedge_{L1}]
{\seq{(\Gamma)^\circledast, (B_1^+)^\bullet \wedge (B_2^+)^\bullet}{(U)^\circledast}}
{\deduce{\seq{(\Gamma)^\circledast, (B_1^+)^\bullet \wedge (B_2^+)^\bullet, (B_1^+)^\bullet}{(U)^\circledast}}{\mathcal D_1}}$

\begin{tabbing}
\qquad \= $\mathcal F_2'$ \= \kill
\>
$\mathcal E_1$ \> ::
 $\ifoc{\Gamma, {\uparrow}(B_1^+ \wedge^+ B_2^+), {\uparrow}B_1^+}{\cdot}{U}$
  \` by i.h.~on $\mathcal D_1$\\
\>
$\mathcal E_1'$ \> ::
 $\ifoc{\Gamma, {\uparrow}(B_1^+ \wedge^+ B_2^+), {\uparrow}B_1^+, {\uparrow}B_2^+}{\cdot}{U}$
  \` by weakening on $\mathcal E_1$\\
\>
$\mathcal E$ \> ::
 $\ifoc{\Gamma, {\uparrow}(B_1^+ \wedge^+ B_2^+), {\uparrow}(B_1^+ \wedge^+ B_2^+)}{\cdot}{U}$\\
  \` by unfocused admissibility lemma $\wedge^+_{uL}$ on $\mathcal E_1'$\\
\>
$\ifoc{\Gamma, {\uparrow}(B_1^+ \wedge^+ B_2^+)}{\cdot}{U}$ \` by contraction on $\mathcal E$
\end{tabbing}

\item[Case] $A^- = B_1^- \wedge^- B_2^-$, \quad
$\mathcal D = \infer[\wedge_{L1}]
{\seq{(\Gamma)^\circledast, (B_1^-)^\bullet \wedge (B_2^-)^\bullet}{(U)^\circledast}}
{\deduce{\seq{(\Gamma)^\circledast, (B_1^-)^\bullet \wedge (B_2^-)^\bullet, (B_1^-)^\bullet}{(U)^\circledast}}{\mathcal D_1}}$

\begin{tabbing}
\qquad \= $\mathcal F_2'$ \= \kill
\>
$\mathcal E_1$ \> ::
 $\ifoc{\Gamma, B_1^- \wedge^- B_2^-, B_1^-}{\cdot}{U}$
  \` by i.h.~on $\mathcal D_1$\\
\> $\mathcal E$ \> ::
$\ifoc{\Gamma, B_1^- \wedge^- B_2^-, B_1^- \wedge^- B_2^-}{\cdot}{U}$\\
  \` by unfocused admissibility lemma $\wedge^-_{uL1}$ on $\mathcal E_1$ \\
\> $\ifoc{\Gamma, B_1^- \wedge^- B_2^-}{\cdot}{U}$
  \` by contraction on $\mathcal E$
\end{tabbing}

\end{description}
There are five other non-initial cases: two corresponding to
$\wedge_{L2}$ that mirror the cases for $\wedge_{L1}$, and three
corresponding to $\bot_L$, $\vee_L$,
and $\supset_L$. All proceed in a similar fashion. 

If our unfocused derivation ends with the ${\it init}$ rule, we must
observe that, according to the definition of erasure, there are four
distinct sequents that all erase to $\seq{\Gamma, p}{p}$.  An atomic
proposition can be the erasure of a positive or negative atomic
proposition, and in suspension-normal sequents atomic
propositions may be suspended or not.  (If $p^+$ is a positive atomic
proposition, then it can appear on the left as either $\susp{p^+}$ or
as ${\uparrow}{\downarrow}\ldots{\uparrow}{\downarrow}{\uparrow}p^+$.)
In each of the four cases, the theorem proceeds directly from shift
removal and the appropriate unfocused admissibility lemma.

This
theorem is named {\tt complete} in the accompanying Twelf development.
\end{proof}

\subsection{Corollaries of focalization}\label{ref:corollaries}

Consider this section a short victory lap. We have established cut
admissibility and identity for the focused sequent calculus, as well
as the focalization and de-focalization properties, without reference
to any properties of the unfocused sequent calculus other than
weakening and exchange. Given these four theorems, the standard
metatheoretic results of the unfocused sequent calculus can be
established as straightforward corollaries. (If our goal was simply to
prove cut admissibility and identity for the unfocused sequent
calculus, then proving focused cut admissibility, identity expansion,
de-focalization, and focalization would admittedly not be the easiest
way to do so!)

The only new thing we need is an arbitrary polarization strategy
$(P)^\circ$ that translates unpolarized propositions to negatively
polarized propositions. It is straightforward to then define
$(\Gamma)^\circ$, the obvious lifting of this function to
contexts.

\begin{corollary}\label{cor:cut}
If $\seq{\Gamma}{P}$ and $\seq{\Gamma, P}{Q}$, then $\seq{\Gamma}{Q}$.
\end{corollary}

\begin{proof}
Since $(-)^\circ$ is defined to be a partial inverse of erasure,
the first given derivation is equally a derivation of
$\seq{((\Gamma)^\circ)^\circledast}{({\downarrow}(P)^\circ)^\circledast}$, and
the second given derivation is equally a derivation of
$\seq{((\Gamma)^\circ,(P)^\circ)^\circledast}{({\downarrow}(Q)^\circ)^\circledast}$. 

By focalization (Theorem~\ref{thm:completeness}), we have the focused
derivations $\ifoc{(\Gamma)^\circ}{\cdot}{{\downarrow}(P)^\circ}$ and
$\ifoc{(\Gamma)^\circ, (P)^\circ}{\cdot}{{\downarrow}(Q)^\circ}$, and
by applying ${\downarrow}_L$ to the second of these derivations we get
$\ifoc{(\Gamma)^\circ}{{\downarrow}(P)^\circ}{{\downarrow}(Q)^\circ}$.
Then, by cut admissibility (Theorem~\ref{thm:cut}, part 4), we obtain
a derivation of $\ifoc{(\Gamma)^\circ}{\cdot}{{\downarrow}(Q)^\circ}$, which by
de-focalization (Theorem~\ref{thm:soundness}) gives us a derivation of
$\seq{((\Gamma)^\circ)^\circledast}{({\downarrow}(Q)^\circ)^\circledast}$, 
which is the same thing as a derivation of $\seq{\Gamma}{Q}$.
\end{proof}

\clearpage
\begin{corollary}\label{cor:identity}
For all $P$, $\seq{\Gamma, P}{P}$.
\end{corollary}

\begin{proof}
By the identity principle, which as discussed is a corollary of
identity expansion (Theorem~\ref{thm:identity}), we can obtain a
derivation of $\ifoc{(\Gamma)^\circ, (P)^\circ}{\cdot}{(P)^\circ}$. By
de-focalization (Theorem~\ref{thm:soundness}), this gives us a
derivation of $\seq{((\Gamma)^\circ,
  (P)^\circ)^\circledast}{((P)^\circ)^\circledast}$, which is the same thing
as a derivation of $\seq{\Gamma,P}{P}$.
\end{proof}

These corollaries ({\tt unfocused-cut} and {\tt unfocused-identity} 
in the accompanying Twelf development) 
are interesting primarily insofar as they establish the 
total dominance that the focused sequent calculus enjoys over 
the unfocused sequent calculus. We have performed precisely
one induction over unpolarized propositions
(implicitly, in the definition of 
$(-)^\circ$) and one induction over unfocused derivations (in the proof of
focalization, Theorem~\ref{thm:completeness}). 
The cut admissibility and identity
expansion lemmas for the focused sequent calculus are strong enough for the
unfocused sequent calculus to inherit its metatheory from the force
of the theorems in the focused setting.

\section{Conclusion}
\label{sec:previous}

We have presented two sequent calculi for different variants of
propositional intuitionistic logic, an unfocused sequent calculus for
unpolarized intuitionistic logic and a focused sequent calculus for
polarized intuitionistic logic.  We then proved a strong theorem about
their equivalence at the level of derivability. The equivalence result
follows from mechanized, structurally inductive proofs establishing
internal soundness and completeness for the focused logic. That
equivalence result implies the internal soundness and completeness of
the unfocused logic. Our systematic approach avoids
tedious invertibility lemmas and allows for a proof, on paper or 
in a mechanized setting, that scales linearly in the number of
connectives and rules.

We will close with a brief survey of existing techniques used to prove
the focalization property, with an emphasis on intuitionistic logic.

\subsection{Comparison to existing focalization proofs}

The most prevalent technique by far has been to do things the long
way.  Andreoli's original presentation of a focused sequent calculus
required a large and tedious series of invertibility lemmas; Andreoli
described these lemmas as ``long but not difficult''
\cite{andreoli92logic}.  Howe's dissertation presents a similar
brute-force approach to the focalization property in the context of
intuitionistic logics, including intuitionistic linear logic
\cite{howe98proof}.  In an unpublished note, Laurent described a
refactored version of the focalization property for classical linear
logic. Laurent staged the proof differently from Andreoli, introducing
several intermediate refinements with some, but not all, of the
restrictions of full focusing.  Laurent's proof is conceptually
clearer than Andreoli's, but it still requires tedious invertibility
lemmas in order to establish the identity property
\cite{laurent04proof}.

The ``grand tour'' strategy of Liang and Miller stands somewhat alone
as an attempt to piggyback on established focusing results, rather
than proving new ones.  Unfocused derivations are translated into
classical linear logic derivations, which are then focused. It is then
only necessary to show that focused derivations can be translated back
out from the focused classical linear logic derivations
\cite{liang09focusing}.  We believe most of instances of this strategy
can be understood, in the context of our system, as specific
polarization strategies, which (as partial inverses of erasure) are
handled generically by our erasure-based proof of focalization.
 


The idea that focalization should arise as a consequence of the 
cut admissibility
and identity properties for a focused logic
originates from Chaudhuri's dissertation
\cite{chaudhuri06focused}. Compared to this work, Chaudhuri's reliance on 
the identity property is less direct, and his proof of identity was 
non-structural, relying on a global decomposition of contexts and
propositions. Chaudhuri's technique was generalized by 
Liang and Miller \shortcite{liang11focused} to any systems meeting 
a general set of criteria; 
these criteria encompass classical and linear logics. In comparison,
the techniques in this paper have not
yet been applied to classical logics,
but have been shown to extend straightforwardly to substructural 
and modal logics \cite{simmons12substructural}.

A line of work by Reed proved focalization by adding extra structure
to the logic being focused. Reed's ``token passing translation''
obtains the necessary structure through the use of linearity and a
distinguished linear atomic proposition \cite{reed08focalizing}. His
work with Pfenning, which was aimed at giving a resource semantics for
substructural logics, obtains the necessary structure through the use
of first-order terms quotiented by an equivalence relation
\cite{reed10focus}. These proofs avoid invertibility lemmas, but their
technique is less direct than ours and may not be as amenable to
formalization in existing logical frameworks.

A wildly different approach to focalization can be found in the
context of Zeilberger's {\it higher-order focusing}
\cite{zeilberger08focusing}. This pattern-based presentation of logic
entirely removes any mention of individual logical connectives from
the core logic; negative and positive propositions are handled in a
completely generic way, in line with synthetic presentations of
focusing. This approach prevents tedious repetition by
default; there aren't enough rules left to tediously induct upon!
Polarization strategy-based focalization for higher-order focusing has
been formalized in the Agda proof assistant, and there do not appear
to be any technical obstacles to mechanizing the erasure-based
approach discussed by Zeilberger \shortcite{zeilberger08unity}.
Higher-order focused proofs represent a significant departure from the
style of presentation in this paper; in particular, higher-order proof
terms are infinitary, which means the Agda mechanization cannot be
ported straightforwardly in Twelf. It is unclear what impact
Zeilberger's strategy of de-functionalizing focused derivations (which
makes them representable in Twelf and, more generally, by
non-infinitary derivations) has on focalization
\cite{zeilberger09defunctionalizing}.

The broad outlines of this paper were first developed in conjunction
with our study of ordered linear logic as a forward chaining logic
programming language \cite{pfenning09substructural}. For the purposes
of that paper, unfocused admissibility in a {\it weakly focused}
sequent calculus -- which did not force invertible rules to be applied
eagerly -- was established the historic (long and tedious) way. A
Twelf proof for weakly focused intuitionistic logic developed at the
same time was the genesis of the structural identity expansion proof
presented here \cite{simmons09weak}.  Eventually, this Twelf proof was
adapted back to ordered linear logic in a technical report that also
introduced the idea of suspended propositions
\cite{simmons11weak}. Unfortunately, to prove full focalization it was
still necessary to prove tedious invertibility lemmas
\cite{simmons11logical}, meaning that the weak focusing technique
gives no advantages beyond those provided by Laurent's refactoring. We
believe this article supersedes our work on weak focusing
entirely. 

Our novel presentation of identity expansion seems to be necessary to
deal with positive propositions. In logics without any interesting
positive structure, simpler techniques have been successfully applied
to prove analogues of the focalization property. The first result in
this line was Miller et al.'s work on {\it uniform proofs} which, like
Andreoli's seminal work, was motivated by logic programming
\cite{miller91uniform}.  We don't intend to fully survey techniques
applicable to settings with only negative connectives, but we will
mention two such systems.  The first system is Jagadeesan et al.'s
$\lambda RCC$, a mixed-paradigm logic programming language with {\it
  atoms} and {\it constraints} that, in retrospect, 
are recognizable as instances of 
negative and positive atoms.
Their focalization proof roughly
resembles the one used by Miller et
al.~\cite{jagadeesan05testing}. The second system is the framework in
which Reed and Pfenning developed their constructive resource
semantics. This system is notable for our purposes because its
focalization proof almost exactly follows our development
\cite{reed10focus}. It was not known at the time how to extend their
proof to a language with non-trivial positive propositions.

\subsection*{Acknowledgments}

Carlo Angiuli, Taus Brock-Nannestad, Illiano Cervesato, 
Kaustuv Chaudhuri, Karl Crary,
Rowan Davies,
Robert Harper, Dan Licata, Chris Martens, Adam Megacz, Dale Miller,
Frank Pfenning, Jason Reed, Fabien Renaud, Bernardo Toninho, 
Sean McLaughlin, Noam Zeilberger, and
two anonymous reviewers offered helpful pointers to existing work
and/or feedback on various drafts of this work.  Frank Pfenning's
insights, particularly his Twelf formulation of identity expansion for
a weakly focused logic (which preceded the formal on-paper formulation
by several years), were particularly invaluable.

Support for this research was provided by the Funda\c{c}\~ao para a
Ci\^encia e a Tecnologia (Portuguese Foundation for Science and
Technology) through the Carnegie Mellon Portugal Program under Grant
NGN-44 and by an X10 Innovation Award from IBM.

\bibliographystyle{acmtrans}
\bibliography{ref}

\end{document}
